\documentclass{article}

\usepackage{amssymb}
\usepackage{amsmath}
\usepackage{latexsym}
\usepackage{amsthm}

\usepackage{booktabs} 

\usepackage{proof}

\usepackage{tikz}
\usepackage{url}

\newcommand{\darrow}{\mathrel{\longrightarrow\!\!\!\!\!\!\!\!\circ\:\:}}

\newcommand{\prarrow}{\mathrel{\ooalign{{$\longleftarrow$}\crcr\hss{$\shortparallel$}\hss}}}
\newcommand{\plarrow}{\mathrel{\ooalign{{$\longrightarrow$}\crcr\hss{$\shortparallel$}\hss}}}

\newtheorem{definition}{Definition}
\newtheorem{lemma}[definition]{Lemma}
\newtheorem{theorem}[definition]{Theorem}
\newtheorem{corollary}[definition]{Corollary}
\newtheorem{proposition}[definition]{Proposition}
\newtheorem{example}[definition]{Example}

\begin{document}

\title{Automated Proofs of Unique Normal Forms\\
w.r.t.\ Conversion for Term Rewriting Systems}

\author{Takahito Aoto\footnote{Niigata University, \texttt{aoto@ie.niigata-u.ac.jp}}
\and 
Yoshihito Toyama\footnote{Tohoku University, \texttt{toyama@riec.tohoku.ac.jp}}}

\maketitle

\begin{abstract}
The notion of normal forms is ubiquitous in various equivalent
transformations.  Confluence (CR), one of the central properties of
term rewriting systems (TRSs), concerns uniqueness of normal
forms. Yet another such property, which is weaker than confluence, is
the property of unique normal forms w.r.t.\ conversion (UNC).  Famous
examples having UNC but not CR include the TRSs consisting of
S,K,I-rules for the combinatory logic supplemented with various
pairing rules (de Vrijer, 1999).  Recently, automated confluence proof
of TRSs has caught attentions leading to investigations of
automatable methods for (dis)proving CR of TRSs; some powerful
confluence tools have been developed as well.  In contrast, there have
been little efforts on (dis)proving UNC automatically yet. Indeed,
there are few tools that are capable of (dis)proving UNC; furthermore,
only few UNC criteria have been elaborated in these tools.  In this
paper, we address automated methods to prove or disprove UNC of given
TRSs.  We report automation of some criteria known so far, and also
present some new criteria and methods for proving or disproving
UNC. Presented methods are implemented over the confluence prover ACP
(Aoto et al., 2009) and an experimental evaluation is reported.
\end{abstract}

\maketitle

\section{Introduction}

The notion of normal forms is ubiquitous in various equivalent transformations---normal 
forms are objects that can not be transformed further.
Two crucial issues arise around the notion of normal forms---one
is whether any object has a normal form and the other is whether they are unique,
so that normal forms can represent the equivalence classes of objects.
The former issue arises various kinds of termination problems.
For the latter, the notion of confluence (CR),
namely that
$s \stackrel{*}{\gets} \circ \stackrel{*}{\to} t$ 
implies $s \stackrel{*}{\to} \circ \stackrel{*}{\gets} t$
for any objects $s,t$, is most well-studied.
Here, $\stackrel{*}{\to}$ is the reflexive transitive closure of 
an equivalent transformation $\to$,
and $\circ$ stands for the composition.
In fact, in the efforts of proving uniqueness of the normal forms,
one encounters the situation of analyzing `local' peaks $s \gets \circ  \to t$,
and then, in order to apply the induction, 
one needs to consider (general) peaks $s \stackrel{*}{\gets} \circ \stackrel{*}{\to} t$.
This naturally leads to the notion of confluence.
In term rewriting, confluence of various systems,  as well as general theories
of confluence for establishing confluence of systems in 
various classes of rewriting systems have been investigated
(see e.g.\ \cite{ToyRTA} for a survey).

Yet another such a property is 
the property of unique normal forms w.r.t.\ conversion (UNC)\footnote{
The uniqueness of normal forms w.r.t.\ conversion is
also often abbreviated as UN in the literature; 
here, we prefer UNC to distinguish it from a similar but different notion of 
unique normal forms w.r.t.\ reduction (UNR), following the convention
employed in CoCo (Confluence Competition).
},
namely that two convertible normal forms are identical,
i.e.\ $s \stackrel{*}{\leftrightarrow} t$  with normal forms $s,t$
implies $s = t$.
Interestingly, \textrm{CR} implies \textrm{UNC},
and this implication is proper, i.e.\ UNC does not imply CR.
Thus, even if the system lacks \textrm{CR}, there still exists a hope that
the system retains \textrm{UNC}.
In term rewriting, famous examples having UNC but 
not CR include term rewriting systems (TRSs) consisting
of S,K,I-rules for the combinatory logic 
supplemented with various pairing rules \cite{KV90,condlin},
whose non-CR have been shown in \cite{Klo80}.
In contrast to CR, the property UNC directly captures the uniqueness of normal forms
in equivalence classes of objects,
which is one of the motivation for verifying CR.
Therefore, it is anticipated that 
many applications would be considered,  once much more powerful
techniques to archieve UNC were obtained.

Compared to CR, however, analyzing UNC is not (yet) very straightforward.
Indeed, not much has been studied on UNC in the field of term 
rewriting---below, we present the short list of known results on UNC in term rewriting.

\begin{proposition}[\cite{UNomega}]
\label{prop:UNomega}
Any non-$\omega$-overlapping TRS has UNC.
\end{proposition}

This recent proposition is,  in fact, an old open problem known as Chew's problem \cite{Chew,ManoOgawaTCS},
and properly generalizes one of the earliest UNC results that
strongly non-overlapping TRSs have UNC \cite{KV90,condlin}.

\begin{proposition}[\cite{MiddeldorpA:phd}]
\label{prop:modular}
UNC is modular for the direct sum.
\end{proposition}

This is one of the earliest results on the modularity of TRSs,
where a property $\varphi$ of TRSs is modular for the direct sum if
$\varphi(\mathcal{R})$ and $\varphi(\mathcal{S})$ implies
$\varphi(\mathcal{R} \cup \mathcal{S})$ for 
TRSs $\mathcal{R}$ and $\mathcal{S}$ over disjoint sets of function symbols.
Modularity holds for some cases having an overlap between sets of function symbols,
namely, for layer-preserving decomposition \cite{AotoTandToyamaY:topdown}
and persistent decomposition \cite{AT97} (we refer to \cite{ACP} for these terminologies).

It is undecidable whether $\mathcal{R}$ has UNC for a given TRS $\mathcal{R}$ in general.
But for some subclasses of TRSs, it is known that it is decidable wheather the given TRS 
in the classes has UNC. 
Concerning the (un)decidability results, we here only present ones on 
the positive side, despite some important negative ones are known as well.

\begin{proposition}[\cite{DecCR}]
\label{prop:dec LLRG}
UNC is decidable for left-linear right-ground TRSs.
\end{proposition}


\begin{proposition}[\cite{UNCshallow}]
\label{prop:dec shallow}
UNC is decidable for shallow TRSs.
\end{proposition}

For the former result, we remark that 
for the class of left-linear right-ground TRSs
first-order theory of rewriting is decidable \cite{DecCR}.
For the latter result, we remark that, in contrast to UNC, 
CR is undecidable for flat TRSs \cite{undecflat},
which is a subclass of shallow TRSs.
Another obvious class for which UNC is decidable is terminating TRSs.

\begin{proposition}[\cite{WeightDec}]
\label{prop:weight decreasing}
Any non-duplicating weight-decreasing joinable TRS has UNC.
\end{proposition}

This criterion is based on a closure condition
of \textit{conditional} critical pairs,
arising from \textit{conditional linearization} of TRSs.
In contrast to various critical pair closure conditions for ensuring confluence
(e.g.\ \cite{Hue80,Dev,Gra96caap,Oku98}),
few such criteria have been known for UNC.

Recently, automated confluence proof of TRSs has caught
attentions leading to investigations of automatable
methods for (dis)proving CR of TRSs;
some powerful confluence tools have been developed as well,
such as ACP \cite{ACP}, CSI \cite{CSI2}, Saigawa \cite{Saigawa} for TRSs, and 
also tools for other frameworks such as conditional TRSs and higher-order TRSs.
This leads to the emergence of the 
Confluence Competition (CoCo)\footnote{\url{http://coco.nue.ie.niigata-u.ac.jp/}},
yearly efforts since 2012.

In contrast, there have been little efforts on (dis)proving UNC
automatically yet. Indeed, there are few tools that are capable of
(dis)proving UNC; furthermore, only few UNC criteria have been elaborated
in these tools.
In CoCo 2017,
the category of UNC runs for the first time\footnote{
As other related properties, the categories of 
the normal form property (NFP)
and that of the uniqueness of normal forms w.r.t.\ reduction (UNR)
have also been run.
Furthermore, there is a combined CR/NFP/UNC/UNR category,
which motivates to shown where the problem lies
at the proper hierarchy of CR $\Rightarrow$ NFP $\Rightarrow$  UNC $\Rightarrow$ UNR.
Among these properties CR and UNC are closed under signature extensions,
but NFP and UNR are not closed under signature extensions,
i.e.\ these properties rely on which set of function symbols
are considered. This motivates us to consider UNC as our current target.
}.
Techniques used by participants are summarized 
as follows:
(1) UNC is decidable for ground TRSs (in polynomial time) \cite{UNground},
(2) UNC is decidable for left-linear right-ground TRSs \cite{DecCR} and 
(3) any non-$\omega$-overlapping TRS has UNC \cite{UNomega}.


In this paper, we address automated methods to prove or disprove UNC.
Main contributions of the paper are summarized as follows.
\begin{itemize}
\item We report new UNC criteria based on the conditional linearization
technique, namely that TRSs have UNC if their conditional linearization is 
parallel-closed or linear strongly closed
(Theorems~\ref{thm:parallel closed} and \ref{thm:strongly closed criteria}).
We also report on automation of these criteria.
Contrast to the earlier result (UNC of strong non-overlapping TRSs) based on 
the conditional linearization technique,
these results are not subsumed by Proposition~\ref{prop:UNomega}.

\item We present a UNC criterion 
which generalizes Proposition~\ref{prop:weight decreasing} given in \cite{WeightDec},
and show how one can effectively check the criterion.
To be more precise on the first item,
we present a critical pair criterion 
ensuring the (abstract) weight-decreasing joinability,
which is slightly general than the one given in \cite{WeightDec}. 

\item We present a novel method, \emph{UNC completion}, for proving and disproving UNC,
and show its correctness (Theorem~\ref{thm:correctness of UNC procedure}).
The method is another application of an abstract UNC principle 
behind the conditional linearization technique. It turns out that the method is much effective 
for proving and disproving UNC of our testbed from Cops (Confluence problems) database,
compared to the conditional linearization approach.

\item We give a transformational method effective for (dis)proving UNC,
named \emph{rule reversing transformation},
and show its correctness (Theorem~\ref{thm:UNC of rule reversing transformation}).
The transformation experimentally turns out to work effectively
when combined with the UNC completion.

\item We present a simple UNC criterion, named  \emph{right-reducibility} (Theorem~\ref{thm:unc by rhs}).

\item We implement UNC criteria except for 
the decidability results (Propositions~\ref{prop:dec LLRG} and \ref{prop:dec shallow}),
and an experimental evaluation is performed on a testbed from Cops database.
Our implementation is built over our confluence prover ACP \cite{ACP}
and is freely available.

\end{itemize}

The rest of the paper is organized as follows.
After introducing necessary notions and notations in Section 2,
we first revisit the conditional linearization technique for proving UNC,
and obtain new UNC criteria based on this approach in Section 3.
In Section 4, 
we present a slightly generalized version of the critical pair criterion
presented in the paper \cite{WeightDec},
and report an automation of the criterion based on Proposition \ref{prop:weight decreasing}. 
In Section 5, we present our novel methods for proving or disproving UNC.
We show an experiment of the presented methods in Section 6, and report
our confluence prover ACP which newly supports UNC (dis)proving in Section 7.
Section 8 concludes.
Most proofs are given in the Appendix.

\section{Preliminaries}

We now fix notions and notations used in the paper. 
We assume familiarity with basic notions in term rewriting (e.g.\ \cite{BaaderFandNipkowT:TR}).

We use $\sqcup$ to denote the multiset union and $\mathbb{N}$ the set of natural numbers.
A sequence of objects $a_1,\ldots,a_n$ is written as $\vec{a}$.
Negation of a predicate $P$ is denoted by $\neg P$.

The composition of relation $R$ and $S$ is denoted by $R \circ S$.
Let $\to$ be a relation on a set $A$.
The reflexive transitive (reflexive, symmetric, equivalent)
closure of the relation $\to$ is denoted by $\stackrel{*}{\to}$
(resp.\ $\stackrel{=}{\to}$, $\leftrightarrow$, $\stackrel{*}{\leftrightarrow}$).
The set NF of \emph{normal forms} w.r.t.\ the relation $\to$
is given by $\textrm{NF} = \{ a \in A \mid a \to b$ for no $b \in A \}$.
The relation $\to$ has \emph{unique normal forms w.r.t.\ conversion} 
(denoted by $\mathrm{UNC}(\to)$)
if $a \stackrel{*}{\leftrightarrow} b$
and $a, b \in \mathrm{NF}$ imply $a = b$.
The relation $\to$ is \emph{confluent} (denoted by $\mathrm{CR}(\to)$)
if ${\stackrel{*}{\gets} \circ \stackrel{*}{\to}}
\subseteq {\stackrel{*}{\to} \circ \stackrel{*}{\gets}}$.
When we consider two relations $\to_1$ and $\to_2$, 
the respective sets of normal forms w.r.t.\ $\to_1$ and $\to_2$
are denoted by $\textrm{NF}_1$ and $\textrm{NF}_2$.
The following proposition, which is proved easily, 
is a basis of the \emph{conditional linearization technique},
which will be used in Sections 3 and 4.
\begin{proposition}[\cite{KV90,condlin}]
\label{prop:abst CR to UNC}
Suppose 
(1) ${\to}_0 \subseteq {\to}_1$,
(2) $\mathrm{CR}(\to_1)$, and 
(3) $\mathrm{NF}_0 \subseteq \mathrm{NF}_1$.
Then, $\text{UNC}(\to_0)$.
\end{proposition}

The set of \emph{terms} over the set $\mathcal{F}$ of arity-fixed \emph{function symbols}
and denumerable set $\mathcal{V}$ of \emph{variables} is denoted by $\mathrm{T}(\mathcal{F},\mathcal{V})$.
The set of variables (in a term $t$) is denoted by $\mathcal{V}$ (resp.\ $\mathcal{V}(t)$).
A term $t$ is \emph{ground} if $\mathcal{V}(t) = \emptyset$.
We abuse the notation $\mathcal{V}(t)$ and denote by $\mathcal{V}(e)$ the set of variables
occurring in any sequence $e$ of expressions.
The \emph{subterm} of a term $t$ at a \emph{position} $p$ is denoted by $t|_p$. 
The root position is denoted by $\epsilon$.
A \emph{context} is a term containing a special constant $\square$ (called \emph{hole}).
If $C$ is a context containing $n$-occurrences of the hole,
$C[t_1,\ldots,t_n]$ denotes the term obtained from $C$ by replacing holes 
with $t_1,\ldots,t_n$ from left to right;
we write $C[t_1,\ldots,t_n]_{p_1,\dots,p_n}$ 
if the occurrences of holes in $C$ are at the positions $p_1,\ldots,p_n$.
For positions $p_1,\ldots,p_n$ in a term $s$,
the expression $s[t_1,\ldots,t_n]_{p_1,\dots,p_n}$ denotes the term obtained from $s$
by replacing subterms at the positions $p_1,\ldots,p_n$ with terms $t_1,\ldots,t_n$ respectively.
We denote by $\vert t \vert_x$ the number of occurrences of a variable $x$ in a term $t$.
Again, we abuse the notation $\vert t \vert_x$ 
and denote by $\vert e \vert_x$ the number of occurrences of a variable $x$
in any sequence of expressions $e$.
A term $t$ is \emph{linear} if $\vert t \vert_x \le 1$ for any $x \in \mathcal{V}(t)$.
A \emph{substitution} $\sigma$ 
is a mapping from $\mathcal{V}$ to $\mathrm{T}(\mathcal{F},\mathcal{V})$
such that the set $\mathrm{dom}(\sigma) = \{ x \in \mathcal{V} \mid \sigma(x) \neq x \}$,
called the \emph{domain} of $\sigma$, is finite.
Each substitution is identified with its homomorphic extension
over $\mathrm{T}(\mathcal{F},\mathcal{V})$.
For simplicity, we often write $t\sigma$ instead of $\sigma(t)$ for substitutions $\sigma$ and terms $t$.
A \emph{most general unifier} $\sigma$ of terms $s$ and $t$ is
denoted by $\mathrm{mgu}(s,t)$.

An \emph{equation} is a pair $\langle l, r \rangle$ of terms,
which is denoted by $l \approx r$. When we indistinguish
lhs and rhs of the equation, we write $l \mathrel{\dot\approx} r$.
We identify equations modulo renaming of variables.
For a set or sequence $\Gamma$ of equations, we denote by $\Gamma\sigma$
the set or the sequence obtained by replacing each equation
$l \approx r$ by $l\sigma \approx r\sigma$.
An equation $l \approx r$
satisfying $l \notin \mathcal{V}$ and $\mathcal{V}(r) \subseteq \mathcal{V}(l)$
is a \emph{rewrite rule} and written as $l \to r$.
A rewrite rule $l \to r$ is \emph{linear}  
if $l$ and $r$ are linear terms; it is \emph{left-linear} (\emph{right-linear})
if $l$ (resp.\ $r$) is a linear term.
A rewrite rule $l \to r$ is \emph{non-duplicating}
if $\vert l \vert_x \ge \vert r \vert_x$ for any $x \in \mathcal{V}(l)$.
A \emph{term rewriting system} (\emph{TRS}, for short) is a finite set of rewrite rules.
A TRS is linear (left-linear, right-linear, non-duplicating) if so are all rewrite rules.
A \emph{rewrite step} of a TRS $\mathcal{R}$ (a set $\Gamma$ of equations)
is a relation $\to_\mathcal{R}$ (resp.\ $\leftrightarrow_\Gamma$) over 
$\mathrm{T}(\mathcal{F},\mathcal{V})$
defined by $s \to_{\mathcal{R}} t$ iff
$s = C[l\sigma]$ and $s = C[r\sigma]$ for some $l \to r \in \mathcal{R}$ (resp.\ $l \mathrel{\dot\approx} r \in \Gamma$)
and context $C$ and substitution $\sigma$.
The position $p$ such that $C|_p = \square$
is called the \emph{redex position} of the rewrite step,
and we sometimes write $s \to_{p,\mathcal{R}} t$ to indicate 
the redex position of this rewrite step explicitly.
A \emph{rewrite sequence} is (finite or infinite) consecutive applications of rewrite steps.
A rewrite sequence of the form 
$t_1 \mathrel{_\mathcal{R}{\gets}} t_0 \to_\mathcal{R} t_2$ is called a \emph{local peak}.

Let $l_1 \to r_1$ and $l_2 \to r_2$ be rewrite rules such that
$\mathcal{V}(l_1) \cap \mathcal{V}(l_2) = \emptyset$.
Suppose that there exists a position $p$ in $l_2$ 
such that $l_2|_p$ and $l_1$ are unifiable.
Let $\sigma = \mathrm{mgu}(l_1,l_2|_p)$.
A local peak $l_2[r_1]_p\sigma \mathrel{_\mathcal{R}{\gets}} l_2\sigma \to_\mathcal{R} r_2\sigma$ is called
a \emph{critical peak} of the rewrite rule $l_1 \to r_1$ over the rewrite rule $l_2 \to r_2$,
provided that it is not the case that $p = \epsilon$ and $l_1 \to r_1$ and $l_2 \to r_2$ are identical.
The term pair $\langle l_2[r_1]_p\sigma, r_2\sigma \rangle$ is called a \emph{critical pair} in $\mathcal{R}$.
It is called an \emph{overlay} critical pair if $p = \epsilon$;
it is called an \emph{inner-outer} critical pair if $p \neq \epsilon$.
The set of (overlay, inner-outer) critical pairs from rules in a 
TRS $\mathcal{R}$ is denoted by $\mathrm{CP}(\mathcal{R})$
(resp.\ $\mathrm{CP}_\textit{out}(\mathcal{R})$, $\mathrm{CP}_\textit{in}(\mathcal{R})$).

Let $l\approx r$ be an equation
and let $\Gamma$ be a sequence $s_1 \approx t_1,\ldots,s_k \approx t_k$ of equations.
An expression of the form $\Gamma \Rightarrow l\approx r$ is called a \emph{conditional equation}.
A conditional equation $\Gamma \Rightarrow l\approx r$ is a conditional rewrite rule
if $l \notin \mathcal{V}$; in this case $\Gamma \Rightarrow l\approx r$ is written
as $l\to r \Leftarrow \Gamma$.
The sequence $\Gamma$ is called the \emph{condition part} of the conditional rewrite rule.
A finite set of conditional rewrite rules is called 
a \emph{conditional term rewriting system} (\emph{CTRS}, for short).
A CTRS is left-linear is so are all rewrite rules.
CTRS $\mathcal{R}$ is said to be of \emph{type 3} (type 1)
if $\mathcal{V}(r) \subseteq \mathcal{V}(l)\cup \mathcal{V}(c)$
(resp.\ $\mathcal{V}(c) \cup \mathcal{V}(r) \subseteq \mathcal{V}(l)$)
for all $l \to r \Leftarrow c \in \mathcal{R}$.

The notion of critical pairs of TRSs is naturally generalized
to the notion of \emph{conditional} critical pairs of CTRSs.
Let $l_1 \to r_1 \Leftarrow \Gamma_1 $ and $l_2 \to r_2 \Leftarrow \Gamma_2$ be conditional
rewrite rules such that
$\mathcal{V}(l_1,r_1,\Gamma_1) \cap \mathcal{V}(l_2,r_2,\Gamma_2) = \emptyset$.
Suppose that $l_2|_p$ and $l_1$ are unifiable and $\sigma = \mathrm{mgu}(l_1,l_2|_p)$.
Then the ternary relation of a sequence of equations and two terms
$\Gamma_1\sigma,\Gamma_2\sigma \Rightarrow  \langle l_2[r_1]_p\sigma, r_2\sigma \rangle$ 
is called a \emph{conditional critical pair},
provided that it is not the case that $p = \epsilon$ and $l_1 \to r_1 \Leftarrow \Gamma_1$ 
and $l_2 \to r_2 \Leftarrow \Gamma_2$ are identical.
Here, $\Gamma_1\sigma,\Gamma_2\sigma$ is a sequence of equations
obtained by the juxtaposition of sequences $\Gamma_1\sigma$ and $\Gamma_2\sigma$.
It is called overlay if $p = \epsilon$;
it is called inner-outer if $p \neq \epsilon$.
The set of conditional critical pairs from conditional rewrite rules in a 
CTRS $\mathcal{R}$ is denoted by $\mathrm{CCP}(\mathcal{R})$
(resp.\ $\mathrm{CCP}_\textit{out}(\mathcal{R})$, $\mathrm{CCP}_\textit{in}(\mathcal{R})$).
A CTRS $\mathcal{R}$ is \emph{orthogonal} if
it is left-linear and $\mathrm{CCP}(\mathcal{R}) = \emptyset$.

Several types of CTRSs are distinguished according to how the
condition part of the conditional rewrite rules is interpreted
to define the rewrite steps.
In this paper, we are interested in \emph{semi-equational} CTRSs where
the equations in condition parts are interpreted by convertibility $\stackrel{*}{\leftrightarrow}$.
Formally, the conditional rewrite step $\to_\mathcal{R}$ of a semi-equational CTRS $\mathcal{R}$
is defined, using auxiliary relations ${\to}_{\mathcal{R}}^{(n)}$ ($n \ge 0$), like this:
\[
\begin{array}{l@{\,}c@{\,}l}
  {\to}_{\mathcal{R}}^{(0)} &=& \emptyset \\[1ex]
  {\to}_{\mathcal{R}}^{(n+1)}&=& 
\{ \langle C[l\sigma],C[r\sigma] \rangle \mid l \to  r \Leftarrow s_1 \approx t_1,\ldots,s_k \approx t_k 
\in \mathcal{R},\\[-1ex]
&& \qquad\qquad\qquad\quad
\forall i ~(1 \le i \le k).\, s_i\sigma \mathrel{\overset{*}{\leftrightarrow}^{\strut}_{\mathcal{R}}{\!\!\!\!\!\!\strut}^{(n)}} t_i\sigma) \} \\
  \to_{\mathcal{R}}&=&\underset{n\in \mathbb{N}}{\bigcup} {\to}^{(n)}_{\mathcal{R}}
\end{array}
\]
The \emph{rank} of conditional rewrite step $s \to_\mathcal{R} t$ is 
the least $n$ such that $s \to^{(n)}_\mathcal{R} t$.

Let $\mathcal{R}$ be a TRS or CTRS.
The set of normal forms w.r.t. $\to_\mathcal{R}$ is written as $\mathrm{NF}(\mathcal{R})$.
A (C)TRS $\mathcal{R}$ has UNC (CR) if $\text{UNC}(\to_\mathcal{R})$ (resp.\ $\text{CR}(\to_\mathcal{R})$)
on the set $\mathrm{T}(\mathcal{F},\mathcal{V})$.
Let $\mathcal{E}$ be a set or sequence of equations or rewrite rules.
We denote $\approx_\mathcal{E}$ the congruence closure of $\mathcal{E}$.
We write $\vdash_\mathcal{E} l \approx r$ if $l \stackrel{*}{\leftrightarrow}_\mathcal{E} r$.
For sets or sequences $\Gamma$ and $\Sigma$ of equations,
we write $\vdash_\mathcal{E} \Sigma$ if $\vdash_\mathcal{E} l \approx r$ for all $l \approx r \in \Sigma$,
and $\Gamma \vdash_\mathcal{E} \Sigma$ if 
$\vdash_\mathcal{E} \Gamma\sigma$ implies  $\vdash_\mathcal{E} \Sigma\sigma$
for any substitution $\sigma$.

\section{Conditional linearization revisited}

The plan of this section is as follows:
We first revisit the conditional linearization technique for proving UNC
in Section \ref{subsec:Conditional linearization}.
Then, we present two new UNC criteria based on this approach 
in Section \ref{subsec:UNC by conditional linearization}.
We remark on automation of check of the criteria in Section \ref{subsec:Automation}.

\subsection{Conditional linearization}
\label{subsec:Conditional linearization}

A conditional linearization is a translation from TRSs to CTRSs
which eliminates non-left-linear rewrite rules, say $f(x,x) \to r$, 
by replacing them with a corresponding conditional rewrite rules,
such as $f(x,y) \to r \Leftarrow x \approx y$.
Formally, let
$l = C[x_1,\dots,x_n]$ with all variable occurrences in $l$ displayed
(i.e.\ $\mathcal{V}(C) = \emptyset$).
Note here $l$ may be a non-linear term
and some variables in $x_1,\ldots,x_n$ may be identical.
Let $l' = C[x_1',\dots,x_n']$ where $x_1',\ldots,x_n'$ are mutually distinct fresh variables
and $\delta$ be a substitution such that $\delta(x_i') = x_i$ $(1 \le i \le n)$
and $\mathrm{dom}(\sigma) = \{ x_1', \ldots, x_n' \}$.
A conditional rewrite rule $l' \to r' \Leftarrow \Gamma$ is 
a \emph{conditional linearization} of a rewrite rule $l \to r$
if $r'\delta = r$
and $\Gamma$ is a sequence of equations of the form $x_i \approx x_j$ ($1 \le i,j \le n$)
such that $x_i' \approx_\Gamma x_j'$ iff $x_i'\delta = x_j'\delta$ holds for any $1 \le i,j \le n$.
A conditional linearization of a TRS 
$\mathcal{R}$ is a semi-equational CTRS (denoted by $\mathcal{R}^L$)
obtained by replacing each rewrite rule with its conditional linearization.
Note that the results of conditional linearizations are not unique,
and any results of conditional linearization is a left-linear CTRS of type 1.

Conditional linearization is useful for showing UNC of non-left-linear TRSs.
The key observation is 
$\mathrm{CR}(\mathcal{R}^L)$
implies 
$\mathrm{UNC}(\mathcal{R})$.
For this, we use Proposition~\ref{prop:abst CR to UNC}
for ${\to}_0  := {\to}_{\mathcal{R}}$
and ${\to}_1  := {\to}_{\mathcal{R}^L}$.
Clearly, ${\to}_{\mathcal{R}} \subseteq {\to}_{\mathcal{R}^L}$,
and thus the condition (1) of Proposition~\ref{prop:abst CR to UNC} holds.
Suppose $\mathrm{CR}(\mathcal{R}^L)$.
Then, one can easily show that 
$\mathrm{NF}(\mathcal{R}) \subseteq \mathrm{NF}(\mathcal{R}^L)$
by induction on the rank of conditional rewrite steps.
Thus, the condition (2) of Proposition~\ref{prop:abst CR to UNC} implies its condition (3).
Hence, $\mathrm{CR}(\mathcal{R}^L)$ implies $\text{UNC}(\mathcal{R})$.

Now, for semi-equational CTRSs, the following confluence criterion is known.

\begin{proposition}[\cite{BergstralJAandKlopJW:crs,ODonnell}]
Orthogonal semi-equational CTRSs are confluent.
\end{proposition}

A TRS $\mathcal{R}$ is \emph{strongly non-overlapping} if $\mathrm{CCP}(\mathcal{R}^L) = \emptyset$.
Hence, it follows:

\begin{proposition}[\cite{KV90,condlin}]
\label{prop:stongly non-overlapping}
Strongly non-overlapping TRSs have UNC.
\end{proposition}

As we mentioned in the introduction,
this proposition is subsumed by Proposition \ref{prop:UNomega}.

\subsection{UNC by conditional linearization}
\label{subsec:UNC by conditional linearization}

We now give some simple extensions of Proposition~\ref{prop:stongly non-overlapping}
which are easily incorporated from \cite{Hue80},
but are not subsumed by Proposition \ref{prop:UNomega}.
For this, let us recall the notion of parallel rewrite steps.
A \emph{parallel rewrite step} $s \plarrow_\mathcal{R} t$ is defined
like this: $s \plarrow_\mathcal{R} t$
iff 
$s = C[l_1\sigma_1,\ldots,l_n\sigma_n]$
and 
$t = C[r_1\sigma_1,\ldots,r_n\sigma_n]$
for some rewrite rules $l_1 \to r_1,\ldots,l_n \to r_n \in \mathcal{R}$
and context $C$ and substitutions $\sigma_1,\ldots,\sigma_n$
($n \ge  0$).
Let us write $\Gamma \vdash_\mathcal{R} u \to v$
if $\vdash_\mathcal{R} \Gamma\sigma$ implies $u\sigma \to_\mathcal{R} v\sigma$ for any substitution $\sigma$.
We define $\Gamma \vdash_\mathcal{R} u \plarrow_\mathcal{R} v$, etc.\ analogously.



The following notion is a straightforward extension of
the corresponding notion of \cite{Hue80,Toyama88}.

\begin{definition}
A semi-equational CTRS $\mathcal{R}$ is \emph{parallel-closed}
if (i) $\Gamma \vdash_\mathcal{R}  u \plarrow v$
for any inner-outer conditional critical pair
$\Gamma \Rightarrow \langle u, v \rangle$ of $\mathcal{R}$,
and
(ii) $\Gamma \vdash_\mathcal{R} u \plarrow \circ \stackrel{*}{\gets} v$
for any overlay conditional critical pair
$\Gamma \Rightarrow \langle u, v \rangle$ of $\mathcal{R}$.
\end{definition}


We now come to our first extension of Proposition~\ref{prop:stongly non-overlapping},
the proof, which is very similar to the one for TRSs, 
is given in Appendix~\ref{sec:omitted proofs}.

\begin{theorem}
\label{thm:parallel closed}
Parallel-closed semi-equational CTRSs are confluent.
\end{theorem}

\begin{corollary}
\label{cor:parallel closed}
A TRS $\mathcal{R}$ has UNC if $\mathcal{R}^L$ is parallel-closed. 
\end{corollary}

Next, we incorporate 
the strong confluence criterion of TRSs \cite{Hue80} to semi-equational CTRSs
in the similar way.

\begin{definition}
A semi-equational CTRS $\mathcal{R}$ is \emph{strongly closed}
if $\Gamma \vdash_\mathcal{R} u \stackrel{*}{\to} \circ \stackrel{=}{\gets} v$
and $\Gamma \vdash_\mathcal{R} u \stackrel{=}{\to} \circ \stackrel{*}{\gets} v$
for any critical pair
$\Gamma \Rightarrow \langle u, v \rangle$ of $\mathcal{R}$.
\end{definition}


Similar to the proof of Theorem~\ref{thm:parallel closed},
the following theorem is obtained in the same way as
in the proof for TRSs.

\begin{theorem}
\label{thm:strongly closed criteria}
Linear strongly closed semi-equational CTRSs are confluent.
\end{theorem}

\begin{corollary}
\label{cor:strongly closed criteria}
A right-linear TRS $\mathcal{R}$ has UNC if $\mathcal{R}^L$ is strongly closed.
\end{corollary}

\begin{example} Let 
\[
\mathcal{R} =
\left\{
\begin{array}{lclcll}
f(x,x,g(y))  &\to& h(y,x) \\
g(a)         &\to& f(a,b,b)\\
h(x,y)       &\to& h(a,y)\\
f(x,x,y)     &\to& h(a,x) \\
\end{array}
\right\}
\]
Since $\mathcal{R}$ is overlapping, not shallow
and not right-ground, 
neither Propositions~\ref{prop:UNomega}, \ref{prop:dec LLRG}
and \ref{prop:dec shallow} apply.
Propositions~\ref{prop:modular}, \ref{prop:weight decreasing} do not apply neither.
By conditional linearization, we obtain
\[
\mathcal{R}^L =
\left\{
\begin{array}{lclcll}
f(x_1,x_2,g(y))  &\to& h(y,x_1) &\Leftarrow& x_1 \approx x_2 & (a)\\
g(a) &\to& f(a,b,b)           & &  & (b)\\
h(x,y)       &\to& h(a,y)    & &  & (c)\\
f(x_1,x_2,y)  &\to& h(a,x_1) &\Leftarrow& x_1 \approx x_2  & (d)\\
\end{array}
\right\}
\]
We have
\[
\mathrm{CCP}_{in}(\mathcal{R}^L) =
\{  x_1 \approx x_2 \Rightarrow  
\langle f(x_1,x_2,f(a,b,b)),  
h(a,x_1)  \rangle \}
\]
and
\[
\mathrm{CCP}_{out}(\mathcal{R}^L) =
\{ x_1 \approx x_2         \Rightarrow  \langle h(a,x_1),   
 h(y,x_1)  \rangle \}.
\]
By $f(x_1,x_2,f(a,b,b)) \to_{\{ (d) \}} h(a,x_1)$
and
$h(a,x_1) \gets_{\{ (c) \}}  h(y,x_1)$,
$\mathcal{R}^L$ is parallel-closed
(or linear strongly closed).
Thus, from Corollary~\ref{cor:parallel closed}
(or Corollary~\ref{cor:strongly closed criteria}),
it follows that $\mathcal{R}$ has UNC.
\end{example}

\begin{figure*}[t]
\begin{center}
\[
\begin{array}{c}
\infer[]
 { \Gamma \sqcup \{ u  \approx v \} \Vdash_{\mathcal{R}} u \sim_0 v}
 {}
\quad
\infer[]
 { \Gamma \Vdash_{\mathcal{R}} t \sim_0 t }
 {}
\quad
\infer[]
 { \Gamma \Vdash_{\mathcal{R}} s \sim_i t }
 { \Gamma \Vdash_{\mathcal{R}} t \sim_i s }
\quad
\infer[]
 { \Gamma \sqcup \Sigma \Vdash_{\mathcal{R}} s \sim_{i+j} u }
 { \Gamma \Vdash_{\mathcal{R}} s \sim_i t  & \Sigma \Vdash_{\mathcal{R}} t \sim_j u }
\\[3ex]
\infer[]
 { \Gamma \Vdash_{\mathcal{R}} C[s] \sim_i C[t] }
 { \Gamma \Vdash_{\mathcal{R}} s \sim_i t }
\quad
\infer[k = \sum_j i_j ]
 {\bigsqcup_j \Gamma_j \Vdash_{\mathcal{R}} \langle u_1,\ldots,u_n \rangle  \sim_k \langle v_1, \ldots,v_n \rangle}
 { \Gamma_1 \Vdash_{\mathcal{R}} u_1 \sim_{i_1} v_1 ~~\cdots~~\Gamma_n \Vdash_{\mathcal{R}} u_n \sim_{i_n} v_n}
\\[3ex]
\infer[]
 { \Gamma \Vdash_{\mathcal{R}} s \sim_i t}
 { \Gamma \Vdash_{\mathcal{R}} s \to_i t}
\quad
\infer[{\scriptsize l \to r \Leftarrow u_1 \approx v_1,\ldots,u_n \approx v_n \in \mathcal{R}}]
 { \Gamma \Vdash_{\mathcal{R}} C[l\sigma] \to_{i+1}  C[r\sigma]}
 { \Gamma \Vdash_{\mathcal{R}} \langle u_1\sigma,\ldots,u_n\sigma \rangle  \sim_i \langle v_1\sigma, \ldots,v_n\sigma \rangle}
\end{array}
\]
\end{center}
\caption{Inference rules for ranked conversions and rewrite steps}
\label{fig:Inference rules}
\end{figure*}



\subsection{Automation}
\label{subsec:Automation}

Even though proofs are rather straightforward,
it is not at all obvious how the conditions
of Theorems~\ref{thm:parallel closed}
and \ref{thm:strongly closed criteria} can be effectively checked.

Let $\mathcal{R}$ be a semi-equational CTRS.
Let $\Gamma \Rightarrow \langle u, v \rangle$ be
an inner-outer conditional critical pair of $\mathcal{R}$,
and consider to check 
$\Gamma \vdash_\mathcal{R}  u \plarrow v$.
For this, we construct the set 
$\mathit{Red} = \{ v' \mid \Gamma \vdash_\mathcal{R}  u \plarrow v' \}$
and check whether $v \in \mathit{Red}$.

To construct the set $\mathit{Red}$,
we seek the possible redex positions in $u$.
Suppose we found conditional rewrite rules $l_1 \to r_1 \Leftarrow \Gamma_1,
l_2 \to r_2 \Leftarrow \Gamma_2 \in \mathcal{R}$
and substitutions $\theta_1,\theta_2$
such that $u = C[l_1\theta_1,l_2\theta_2]$.
Then we obtain $u \plarrow C[r_1\theta_1,r_2\theta_2]$ if 
$\vdash_\mathcal{R} \Gamma_1\theta_1$
and 
$\vdash_\mathcal{R} \Gamma_2\theta_2$,
i.e.\ 
$s \stackrel{*}{\leftrightarrow}_\mathcal{R} t$
for any equations $s \approx t$  in $\Gamma_1\theta_1 \cup \Gamma_2\theta_2$.
Now, for checking $\Gamma \vdash_\mathcal{R}  u \plarrow v$,
it suffices to consider the case $\vdash_\mathcal{R} \Gamma$ holds.
Thus, we may assume  $s' \stackrel{*}{\leftrightarrow}_\mathcal{R} t'$
for any $s' \approx t'$  in $\Gamma$.
Therefore, the problem is to check 
whether $s' \stackrel{*}{\leftrightarrow}_\mathcal{R} t'$ for $s' \approx t'$  in $\Gamma$
implies 
$s \stackrel{*}{\leftrightarrow}_\mathcal{R} t$
for any equations $s \approx t$  in $\Gamma_1\theta_1 \cup \Gamma_2\theta_2$.

To check this, we use the following sufficient condition:
$s \approx_\Gamma t$ for all $s \approx t \in \Gamma_1\theta_1 \cup \Gamma_2\theta_2$.
Note there $\approx_\Gamma$ is the congruence closure of $\Gamma$.
Since congruence closure of a finite set of equations is
decidable \cite{BaaderFandNipkowT:TR}, this approximation is indeed automatable.

\begin{example}
Let
\[
\mathcal{R} =
\left\{
\begin{array}{lclcll}
P(Q(x)) &\to& P(R(x)) &\Leftarrow& x \approx A & (a)\\
Q(H(x)) &\to& R(x)     &\Leftarrow& S(x) \approx H(x) & (b)\\
R(x) &\to& R(H(x))     &\Leftarrow& S(x) \approx A & (c)\\
\end{array}
\right\}
\]
Then we have
$\mathrm{CCP}(\mathcal{R})
= \mathrm{CCP}_{in}(\mathcal{R}) =$
\[
\{ S(x) \approx H(x), H(x) \approx A \Rightarrow \langle P(R(x)), P(R(H(x))) \rangle \}
\]
Now, in order to apply rule $(c)$ to have $P(R(x)) \plarrow_\mathcal{R} P(R(H(x)))$,
we have to check the condition 
$S(x) \stackrel{*}{\leftrightarrow}_\mathcal{R} A$.
This holds,
since 
we can suppose $S(x) \stackrel{*}{\leftrightarrow}_\mathcal{R} H(x)$
and $H(x) \stackrel{*}{\leftrightarrow}_\mathcal{R} A$.
This is checked by 
$S(x)  \approx_\Sigma A$,
where $\Sigma = \{ S(x) \approx H(x), H(x) \approx A \}$.
\end{example}

\section{Automating UNC proof of non-duplicating TRSs}

In this section, we show a slight generalization of
the UNC criterion based on Proposition~\ref{prop:weight decreasing} \cite{WeightDec},
and show how the criterion can be decided.
First, we briefly capture necessary notions and notations from the paper \cite{WeightDec}.

A \emph{left-right separated (LR-separated) conditional rewrite rule}
is $l \to r \Leftarrow x_1 \approx y_1,\ldots,x_n \approx y_n$
such that 
(i) $l \notin \mathcal{V}$ is linear, 
(ii) $\mathcal{V}(l) =  \{ x_i \}_i$ and $\mathcal{V}(r) \subseteq  \{ y_i \}_i$
(iii) $\{ x_i \}_i \cap\{ y_i  \}_i = \emptyset$, and 
(iv) $x_i \neq x_j$ for $i \neq j$.
Here, note that some variables in $y_1,\ldots,y_n$ can be identical.
A finite set of LR-separated conditional rewrite rules is called
an \emph{LR-separated conditional term rewriting system} (\emph{LR-separated CTRS}, for short).
An LR-separated conditional rewrite rule 
$l \to r \Leftarrow x_1 \approx y_1,\ldots,x_n \approx y_n$
is \emph{non-duplicating} if 
$|r|_y \le |y_1,\ldots,y_n|_y$ for all  $y \in \mathcal{V}(r)$.

The LR-separated conditional linearization 
translated TRSs to LR-separated CTRSs.
This is given as follows:
Let $C[y_1,\ldots,y_n] \to r$ be a rewrite rule,
where $\mathcal{V}(C)= \emptyset$.
Here, some variables in $y_1,\ldots,y_n$ may be identical.
Then, we take fresh distinct $n$ variables $x_1,\ldots,x_n$,
and put $C[x_1,\ldots,x_n] \to r \Leftarrow x_1 \approx y_1,\ldots, x_n \approx y_n$
as the result of the translation.
It is easily seen that the result is indeed 
an LR-separated conditional rewrite rule.
It is also easily checked that if the rewrite rule is 
non-duplicating then so is the result of the translation 
(as an LR-separated conditional rewrite rule).
The LR-separated conditional linearization $\mathcal{R}^{LRS}$ of a TRS $\mathcal{R}$
is obtained by applying the translation to each rule.

It is shown in \cite{WeightDec} that semi-equational non-duplicating 
LR-separated CTRSs are confluent if their conditional critical pairs
satisfy some closure condition, which makes the rewrite steps `weight-decreasing joinable'.
By applying the criterion to LR-separated conditional linearization of TRSs, 
they obtained a criterion of \textrm{UNC} for non-duplicating TRSs.
%
Note that rewriting in LR-separated CTRSs is (highly) non-deterministic;
even reducts of rewrite steps at the same position by the same rule is generally not unique,
not only reflecting semi-equational evaluation of the conditional part
but also by the $\mathcal{V}(l) \cap \mathcal{V}(r) = \emptyset$ for 
LR-separated conditional rewrite rule $l \to r \Leftarrow c$.
Thus, how to effectively check the sufficient condition of weight-decreasing joinability 
is not very clear, albeit it is mentioned in \cite{WeightDec} that the decidability is clear.

For obtaining an algorithm for computing the criterion,
we introduce ternary relations
parameterized by an LR-separated CTRS $\mathcal{R}$ and $n \in \mathbb{N}$
as follows.

\begin{definition}
\label{def:derivation of sim}
The derivation rules for $\Gamma \Vdash_{\mathcal{R}} u \sim_n v$ 
and $\Gamma \Vdash_{\mathcal{R}} u \to_n v$ are given in Figure~\ref{fig:Inference rules}.
Here, $n \in \mathbb{N}$ and $\Gamma$ is a multiset of equations.
\end{definition}

Intuitively, $\Gamma \Vdash_{\mathcal{R}} u \sim_n v$ means 
that $u \stackrel{*}{\leftrightarrow}_\mathcal{R} v$ using the assumption $\Gamma$
where the number of rewrite steps is $n$ in total (i.e.\ including those 
used in checking conditions).
Main differences to the relation 
$\underset{\Gamma}{\sim}$ in \cite{WeightDec} are twofold:
\begin{enumerate}
\item Instead of considering a special constant $\bullet$, we use an index of natural number.
The number of $\bullet$ corresponds to the index number.
\item Auxiliary equations in $\Gamma$ are allowed in our notation of $\Gamma \Vdash_{\mathcal{R}} u \sim_n v$.
On the contrary, $\Gamma$ in $\underset{\Gamma}{\sim}$ in \cite{WeightDec} does not
allow auxiliary equations in $\Gamma$
\end{enumerate}
The former is rather a notational convenience; however, this is useful
to designing the effectiv procedure to check the UNC criteria presented below.
The latter is convenient to prove the satisfiability of constraints
on such expressions.
We refer 
to  Appendix~\ref{sec:Comparison to WeightDec} for more precise comparison with \cite{WeightDec}.

The following is a slight generalization of the main result of \cite{WeightDec}.
A proof is given in Appendix~\ref{sec:Comparison to WeightDec}.

\begin{theorem}
\label{thm:I}
A semi-equational non-duplicating LR-separated CTRS $\mathcal{R}$ is 
weight-decreasing joinable
if for any critical pair $\Gamma \Rightarrow  \langle s,t \rangle$ of $\mathcal{R}$,
either
(i) $\Gamma \Vdash_{\mathcal{R}} s \sim_{\le 1} t$,
(ii) $\Gamma \Vdash_{\mathcal{R}} s  \leftrightarrow_2 t$, or 
(iii) $\Gamma \Vdash_{\mathcal{R}} s  \to_i \circ \sim_j t$ with $i + j \le 2$ and 
$\Gamma \Vdash_{\mathcal{R}} t  \to_{i'} \circ \sim_{j'} s$ with $i' + j' \le 2$.
\end{theorem}

Thus, any non-duplicating TRS $\mathcal{R}$ has UNC 
if all CCPs of $\mathcal{R}^{LRS}$ satisfy some of conditions (i)--(iii).

Thanks to our new formalization of sufficient condition,
decidability of the condition follows.

\begin{theorem}
\label{thm:decidability of WDJ}
The condition of Theorem~\ref{thm:I} is decidable.
\end{theorem}

\begin{proof}
We show that each condition (i)--(iii) is decidable.
Let $\Gamma$ be a (finite) multiset of equations,
$s,t$ terms, and $\vec{s},\vec{t}$ sequences of terms.
The claim follows by showing the following series of sets are finite and effectively constructed one by one:
(a)
$\mathrm{SIM}_0(\Gamma, s) =  \{ \langle\Sigma, t\rangle  \mid \Gamma{\setminus}\Sigma \Vdash_{\mathcal{R}} s \sim_0 t  \}$,
(b)
$\mathrm{SIM}_0(\Gamma, \vec s\,) =  \{ \langle\Sigma, \vec t\, \rangle  \mid \Gamma{\setminus}\Sigma \Vdash_{\mathcal{R}} \vec s\,  \sim_0 \vec t\,  \}$,
(c)
$\mathrm{RED}_1(\Gamma, s,t) =  \{ \Sigma  \mid \Gamma{\setminus}\Sigma \Vdash_{\mathcal{R}} s \to_1 t  \}$,
(d)
$\mathrm{SRS}_{010}(\Gamma, s, t) =  \{ \Sigma  \mid \Gamma{\setminus}\Sigma \Vdash_{\mathcal{R}} s \sim_0 \circ \to_1 \circ \sim_0 t  \}$,
(e)
$\mathrm{SIM}_1(\Gamma, s, t) =  \{ \Sigma  \mid \Gamma{\setminus}\Sigma \Vdash_{\mathcal{R}} s \sim_1 t  \}$,
(f)
$\mathrm{SIM}_1(\Gamma, \vec s, \vec t\,) =  \{ \Sigma \mid \Gamma{\setminus}\Sigma \Vdash_{\mathcal{R}} \vec s\,  \sim_1 \vec t\,  \}$, and
(g)
$\mathrm{RED}_2(\Gamma, s,t) =  \{ \Sigma \mid \Gamma{\setminus}\Sigma \Vdash_{\mathcal{R}} s \to_2 t  \}$.
\end{proof}

\begin{example} Let 
\[
\mathcal{R} =
\left\{
\begin{array}{lclcll}
f(x,x)      &\to& h(x,f(x,b))    \\
f(g(y),y)    &\to& h(y,f(g(y),c(b))) \\
h(c(x),b)    &\to& h(b,b) \\
c(b)         &\to& b \\
\end{array}
\right\}
\]
Since $\mathcal{R}$ is overlapping, not right-ground, and not shallow,
Propositions~\ref{prop:UNomega}, \ref{prop:dec LLRG}, \ref{prop:dec shallow}
do not apply. 
Proposition~\ref{prop:modular} and Theorems~\ref{thm:parallel closed}, \ref{thm:strongly closed criteria}
do not apply either.
By conditional linearization, we obtain
$\mathcal{R}^{LRS} =$
\[
\left\{
\begin{array}{lclcll}
f(x_1,x_2)  &\to& h(x,f(x,b))  &\Leftarrow& x_1 \approx x, x_2  \approx x \\
f(g(y_1),y_2)  &\to& h(y,f(g(y),c(b)))  &\Leftarrow& y_1 \approx y, y_2  \approx y \\
h(c(x),b)    &\to& h(b,b) \\
c(b)          &\to& b     \\
\end{array}
\right\}
\]
We have an overlay critical pair
\[
\left\{
\begin{array}{cc}
y_1 \approx y & (a)\\
y_2 \approx y & (b)\\
g(y_1) \approx x & (c)\\
y_2 \approx x  & (d)\\
\end{array}
\right\}
\Rightarrow  \langle h(x,f(x,b)),  
h(y,f(g(y),c(b)))  \rangle
\]
(Another one is its symmetric version.)
Let $\Gamma = \{ (a),(b),(c),(d) \}$,
$s = h(y,f(g(y),c(b)))$
and $t = h(x,f(x,b)))$.
To check the criteria of Theorem~\ref{thm:I},
we start computing 
$\mathrm{SIM}_0(\Gamma, s)$
and
$\mathrm{SIM}_0(\Gamma, t)$.
For example, the former equals to 
\[
\left\{
\begin{array}{l}
\langle \{ (a),(b),(c),(d) \}, h(y,f(g(y),c(b))) \rangle\\
\langle \{ (b),(c),(d) \}, h(y_1,f(g(y),c(b))) \rangle\\
\langle \{ (b),(c),(d) \}, h(y,f(g(y_1),c(b))) \rangle\\
\langle \{ (b),(d) \}, h(y,f(x,c(b))) \rangle\\
\langle \{ (a),(c),(d) \}, h(y_2,f(g(y),c(b))) \rangle\\
\langle \{ (a),(c),(d) \}, h(y,f(g(y_2),c(b))) \rangle\\
\langle \{ (a),(c) \}, h(x,f(g(y),c(b))) \rangle\\
\langle \{ (a),(c) \}, h(y,f(g(x),c(b))) \rangle\\
\langle \{ (c),(d) \}, h(y_1,f(g(y_2),c(b))) \rangle\\
\langle \{ (c),(d) \}, h(y_2,f(g(y_1),c(b))) \rangle\\
\langle \{ (c) \}, h(y_1,f(g(x),c(b))) \rangle\\
\langle \{ (c) \}, h(x,f(g(y_1),c(b))) \rangle\\
\langle \{ (d) \}, h(y_2,f(x,c(b))) \rangle\\
\langle \emptyset, h(x,f(x,c(b))) \rangle
\end{array}
\right\}.
\]
We now can check  $s \sim_0 t$ does not hold
by $\langle \Gamma', t\rangle \in \mathrm{SIM}_0(\Gamma, s)$
for no $\Gamma'$.
To check $\Gamma \Vdash s \to_1 t$,
we compute $\mathrm{RED}_1(\Gamma,s,t)$.
For this, we check there exist a context $C$ and substitution $\theta$
and rule $l \to r \Leftarrow \Gamma \in \mathcal{R}^{LRS}$
such that $s = C[l\theta]$ and $t = C[r\theta]$.
In our case, it is easy to see $\mathrm{RED}_1(\Gamma,s,t) = \emptyset$.
Next to check $\Gamma \Vdash s \sim_1 t$,
we compute $\mathrm{SRS}_{010}(\Gamma, s, t)$.
This is done by, for each $\langle \Gamma', s' \rangle  \in \mathrm{SIM}_0(\Gamma, s)$,
computing $\langle \Sigma, t' \rangle  \in \mathrm{SIM}_0(\Gamma', t)$
and check there exists $\Sigma \in \mathrm{RED}_1(\Sigma',s',t')$.
In our case,
for $\langle \emptyset, h(x,f(x,c(b))) \rangle \in \mathrm{SIM}_0(\Gamma, s)$
we have  $\langle \emptyset, t \rangle \in \mathrm{SIM}_0(\emptyset, t)$,
and $\emptyset \in \mathrm{RED}_1(\emptyset, h(x,f(x,c(b))), t)$.
Thus, we know $h(x,f(x,c(b))) \to_1 h(x,f(x,b))$.
Hence, for these overlay critical pairs, we have
$y_1 \approx y, y_2 \approx y, g(y_1) \approx x, y_2 \approx x 
\Vdash_{\mathcal{R}} 
h(y,f(g(y),c(b))) \sim_1 h(x,f(x,b))$.
We also have $\mathrm{CCP}_{in}(\mathcal{R}^{LRS}) = \{$
$\emptyset \Rightarrow  \langle h(b,b),  
h(b,b) \rangle \}$.
For this inner-outer critical pair, 
it follows that $\Vdash_{\mathcal{R}} h(b,b) \sim_0 h(b,b)$
using $\langle \emptyset, h(b,b) \rangle \in \mathrm{SIM}_0(\emptyset, h(b,b))$.
Thus, from Theorem~\ref{thm:I}, $\mathcal{R}^{LRS}$ is weight-decreasing.
Hence, it follows $\mathcal{R}$ has UNC.
We remark that, in order to derive $\Vdash_{\mathcal{R}} h(b,b) \sim_0 h(b,b)$,
we need the reflexivity rule.
However, since the corresponding Definition of $\sim$ in the paper \cite{WeightDec}
lacks the reflexivity rule, 
the condition of weight-decreasing in \cite{WeightDec} (Definition 9) does not hold for $\mathcal{R}^{LRS}$.
A part of situations where the reflexivity rule is required is, however, covered by the congruence rule;
thus the reflexivity rule becomes necessary when there exists a trivial critical pair such as above.
\end{example}

\begin{figure*}[t]
\noindent
\begin{flushleft}
\hspace*{.2cm}\textbf{Input}: TRS $\mathcal{R}$, predicates $\varphi, \Phi$\\
\hspace*{.2cm}\textbf{Output}: UNC or NotUNC or Failure (or may diverge)\\
\end{flushleft}
\begin{enumerate}
\item[\textbf{Step 1.}]
Compute the set $\mathrm{CP}(\mathcal{R})$ of critical pairs of $\mathcal{R}$.

\item[\textbf{Step 2.}]
If $\Phi(u, v)$ for all $\langle u, v \rangle \in \mathrm{CP}(\mathcal{R})$
and $\varphi(\mathcal{R})$ then return UNC.

\item[\textbf{Step 3.}]
Let $\mathcal{S} := \emptyset$.
For each $\langle u, v \rangle \in \mathrm{CP}(\mathcal{R})$ with $u \neq v$ for which $\Phi(u, v)$ does not hold,
do:
\begin{enumerate}
\item
If $u,v \in \mathrm{NF}(\mathcal{R})$, then exit with NotUNC.

\item If $u \notin \mathrm{NF}(\mathcal{R})$ and $v \in \mathrm{NF}(\mathcal{R})$,
then if $\mathcal{V}(v) \not\subseteq \mathcal{V}(u)$ 
then exit with NotUNC,
otherwise update $\mathcal{S} :=  \mathcal{S} \cup \{ u \to v \}$.

\item If $v \notin \mathrm{NF}(\mathcal{R})$ and $u \in \mathrm{NF}(\mathcal{R})$,
then if $\mathcal{V}(u) \not\subseteq \mathcal{V}(v)$ 
then exit with NotUNC,
otherwise update $\mathcal{S} :=  \mathcal{S} \cup \{ v \to u \}$.

\item If $u,v \notin \mathrm{NF}(\mathcal{R})$
then find $w$ such that $u \stackrel{*}{\to}_\mathcal{R} w$ ($v \stackrel{*}{\to}_\mathcal{R} w$),
and $\mathcal{V}(w) \subseteq \mathcal{V}(v)$ (resp.\ $\mathcal{V}(w) \subseteq \mathcal{V}(v)$).
If it succeeds then update $\mathcal{S} :=  \mathcal{S} \cup \{ v \to w \}$.
\end{enumerate}

\item[\textbf{Step 4.}]
If $\mathcal{S} = \emptyset$ then return Failure; 
otherwise update $\mathcal{R} := \mathcal{R} \cup \mathcal{S}$
and go back to Step~1.\\[-4ex]
\end{enumerate}
\caption{UNC completion procedure parameterized by predicates $\varphi,\Phi$}
\label{fig:UNC completion}
\end{figure*}

\section{UNC completion and other methods}

In this section, we present some new approaches
for proving and disproving UNC.

Firstly, observe that the conditional linearization does not change the input 
TRSs if they are left-linear. Thus, 
the technique has no effects on left-linear rewrite rules. 
But, as one can easily see, however, it is not at all guaranteed that left-linear TRSs have UNC.


Now, observe that a key idea in the conditional linearization technique is 
that CR of an approximation of a TRS implies UNC of the original TRS.
The first method presented in this section 
is based on the observation that one can 
also use the approximation other than conditional linearization.
To fit our usage, we now slightly modify Proposition~\ref{prop:abst CR to UNC}.

\begin{lemma}
\label{lem:approximation II}
Suppose
(1) ${\to}_\mathcal{R} \subseteq {\to}_{\mathcal{S}} \subseteq {\stackrel{*}{\leftrightarrow}}_\mathcal{R}$ and 
(2) $\mathrm{NF}(\mathcal{R}) \subseteq \mathrm{NF}(\mathcal{S})$.
Then,
(i) If $\mathrm{CR}(\mathcal{S})$ then $\mathrm{UNC}(\mathcal{R})$.
(ii) If there exists distinct $s,t \in \mathrm{NF}(\mathcal{S})$
such that $s \stackrel{*}{\leftrightarrow}_\mathcal{S} t$,
then $\neg \mathrm{UNC}(\mathcal{R})$.
\end{lemma}

Our approximation $\mathcal{S}$ of a TRS $\mathcal{R}$
is given by adding auxiliary rules aiming to obtain $\text{CR}$ of the TRS $\mathcal{S}$, 
in such a way that conditions (1) and (2) of the lemma are guaranteed.

\begin{definition}
A UNC completion procedure is given as Figure~\ref{fig:UNC completion}.
Its input are a TRS and two predicates $\varphi,\Phi$ such
that for any TRS $\mathcal{S}$ satisfying $\varphi(\mathcal{S})$
if $\Phi(u, v)$ for all critical pairs $\langle u, v \rangle$ of $\mathcal{S}$,
then $\mathrm{CR}(\mathcal{S})$.
\end{definition}

\begin{example}[Cops $\sharp$254]
Let
\[
\mathcal{R} =
\left\{
\begin{array}{lclcll}
a &\to& f(c)\\
a &\to& f(h(c))\\
f(x) &\to& h(f(x))\\
\end{array}
\right\}
\]
Since $\mathcal{R}$ is overlapping, not right-ground, and not shallow,
Propositions~\ref{prop:UNomega}, \ref{prop:dec LLRG}, \ref{prop:dec shallow}
do not apply. 
Proposition~\ref{prop:modular} does not apply either.
Now, 
let us apply the UNC completion procedure to $\mathcal{R}$
using linear strongly closed criteria for confluence.
For this, take $\varphi(\mathcal{R})$ as $\mathcal{R}$ is linear,
and $\Phi(u,v)$ as
$(u \stackrel{*}{\to} \circ \stackrel{=}{\gets} v)
\land (u \stackrel{=}{\to} \circ \stackrel{*}{\gets} v)$.
In Step 3, we find an overlay critical pair $\langle f(h(c)), f(c) \rangle$,
for which $\Phi$ is not satisfied.
Since $f(h(c))$ and $f(c)$ are not normal, we go to Step 3(b).
Take $w := f(c)$ and add a rewrite rule $f(h(c)) \to f(c)$
to obtain $\mathcal{R} := \mathcal{R} \cup \{  f(h(c)) \to f(c) \}$.
Now, the updated $\mathcal{R}$ is linear and strongly closed
(and thus, $\mathcal{R}$ is confluent).
Hence, the procedure returns UNC at Step 2.
\end{example}

We now prove the correctness of the procedure.
We first present two simple lemmas for this.

\begin{lemma}
\label{lem:addition}
Suppose
$l \stackrel{*}{\leftrightarrow}_\mathcal{R} r$,
$l \notin \mathrm{NF}(\mathcal{R})$,
and $l \to r$ is a rewrite rule.
Then, $\mathrm{UNC}(\mathcal{R})$
iff 
$\mathrm{UNC}(\mathcal{R} \cup \{ l \to r \})$.
\end{lemma}

\begin{lemma}
\label{lem:disproving by eliminating variable}
Suppose
$s \stackrel{*}{\leftrightarrow}_\mathcal{R} t$, $t \in \mathrm{NF}(\mathcal{R})$
and $\mathcal{V}(t) \not\subseteq \mathcal{V}(s)$.
Then $\neg \mathrm{UNC}(\mathcal{R})$.
\end{lemma}

\begin{theorem}
\label{thm:correctness of UNC procedure}
The UNC completion procedure is correct, i.e.\
if the procedure returns UNC then $\mathrm{UNC}(\mathcal{R})$,
and if the procedure returns NotUNC then $\neg \mathrm{UNC}(\mathcal{R})$.
\end{theorem}

We now present two simple results, which turn out effective
for some examples.

\begin{definition}
Let $\mathcal{R}$ be a TRS.
We write $\mathcal{R} \leadsto \mathcal{R}'$
if $\mathcal{R}' = (\mathcal{R} \setminus \{ l \to r \}) \cup  \{ l \to l, r \to l \}$
for some $l \to r \in \mathcal{R}$ such that
$r \notin \mathrm{NF}(\mathcal{R})$ and $r \to l$ is a rewrite rule,
or $\mathcal{R}' = \mathcal{R} \setminus \{ l \to r \}$
for some $l \to r \in \mathcal{R}$ such that
$l = r$ and $l \notin \mathrm{NF}(\mathcal{R} \setminus \{ l \to r \})$.
Any transformation $\mathcal{R} \stackrel{*}{\leadsto} \mathcal{R}'$
is called a \emph{rule reversing transformation}.
\end{definition}

\begin{theorem}
\label{thm:UNC of rule reversing transformation}
Let $\mathcal{R}'$ be a TRS obtained by 
a rule reversing transformation from $\mathcal{R}$.
Then, $\mathrm{UNC}(\mathcal{R})$ iff $\mathrm{UNC}(\mathcal{R}')$.
\end{theorem}

\begin{example}
Let
\[
\mathcal{R} =
\left\{
\begin{array}{lclcll}
   a &\to& f(a)\\
   h(c,a) &\to& b\\
   h(a,x) &\to& h(x,f(x))\\
\end{array}
\right\}
\]
Since $\mathcal{R}$ is overlapping and not shallow,
Propositions~\ref{prop:UNomega}, \ref{prop:dec shallow}
do not apply. Proposition~\ref{prop:modular} 
does not apply either.
Since it is left-linear, conditional linearization technique
does not apply.
Note here that $f(a) \notin \mathrm{NF}(\mathcal{R})$ because of the rule $a \to f(a)  \in \mathcal{R}$.
Thus, one can apply the rule reversing transformation
to obtain 
\[
\mathcal{R}' =
\left\{
\begin{array}{lclcll}
   a &\to& a\\
   f(a) &\to& a\\
   h(c,a) &\to& b\\
   h(a,x) &\to& h(x,f(x))\\
\end{array}
\right\}
\]
Now, it is easy to check $\mathcal{R}'$ is left-linear and development closed,
and thus $\mathcal{R}'$ is confluent.
Thus, from Theorem~\ref{thm:UNC of rule reversing transformation},
we conclude $\mathcal{R}$ has UNC.
\end{example}

\begin{definition}
A TRS $\mathcal{R}$ is said to be \emph{right-reducible}
if $r \notin \mathrm{NF}(\mathcal{R})$ for all $l \to r \in \mathcal{R}$.
\end{definition}

\begin{theorem}
\label{thm:unc by rhs}
Any right-reducible TRS has $\mathrm{UNC}$.
\end{theorem}

\begin{example}[Cops $\sharp$126]
\[
\mathcal{R} =
\left\{
\begin{array}{lclcll}
 f(f(x,y),z) &\to& f(f(x,z),f(y,z))
\end{array}
\right\}
\]
The state of the art confluence tools fail to
prove confluence of this example.
However, it is easy to see $\mathcal{R}$ is right-reducible,
and thus, UNC is easily obtained automatically.
\end{example}

\begin{table*}\small
\begin{center}
\begin{tabular}{l|c|c|c|c|c|c|c|c|c|c}\hline
~~~~\textit{without \textbf{(rev)}}   & \textbf{(sno)}& \textbf{($\omega$)}&\textbf{(pcl)}&\textbf{(scl)}&\textbf{(wd)}&\textbf{(sc)}~1/2/3&\textbf{(dc)}~1/2/3&\textbf{(rr)}&\textbf{(cp)} & \textit{all} \\\hline
\texttt{YES}&        7      &     7   &  7   & 0   &    2        &  0/6/8      &   0/6/9        &   35        &   0          & 47 \\
\texttt{NO} &        0      &     0   &  0   & 0   &    0        &  14/33/41   &   14/34/41     &   0         &   42         & 58 \\
\texttt{YES}+\texttt{NO} 
            &        7      &     7   &  7   & 0   &    2        &  14/39/49   &   14/40/50     &   35        &   42          & 105 \\
timeout (60s)&        0     &     0   &  7   & 0   &    0        &  2/7/17     &   4/10/19      &    0        &   0          & -- \\\hline
\multicolumn{10}{c}{}\\[-1ex]\hline
~~~~\textit{with \textbf{(rev)}}   & \textbf{(sno)}& \textbf{($\omega$)}&\textbf{(pcl)}&\textbf{(scl)}&\textbf{(wd)}&\textbf{(sc)}~1/2/3&\textbf{(dc)}~1/2/3&\textbf{(rr)}&\textbf{(cp)} & \textit{all} \\\hline
\texttt{YES}&        3      &     3   &  3   & 0   &    0        &  24/42/45   &   24/35/39     &  35         &   0          & 62 \\
\texttt{NO} &        0      &     0   &  0   & 0   &    0        &  15/39/44   &   15/40/44     &   0         &   35         & 56 \\
\texttt{YES}+\texttt{NO} 
            &        3      &     3   &  3   & 0   &    0        &  39/81/89   &   39/75/83     &   35        &   35         & 118  \\
timeout (60s)&       0      &     0   &  0   &  0  &    2        &  3/4/8      &   3/5/9        &   0         &   0          & -- \\\hline
\multicolumn{10}{c}{}\\[-1ex]\hline
~~~~\textit{both}   & \textbf{(sno)}& \textbf{($\omega$)}&\textbf{(pcl)}&\textbf{(scl)}&\textbf{(wd)}&\textbf{(sc)}~1/2/3&\textbf{(dc)}~1/2/3&\textbf{(rr)}&\textbf{(cp)} & \textit{all} \\\hline
\texttt{YES}+\texttt{NO} 
            &        7      &     7   &  7  & 0  &    2        &  39/82/90    &   39/78/85      &   35        &   42          & 127 \\\hline
\end{tabular}
\end{center}
\caption{Test on basic criteria}
\label{table:test I}
\end{table*}

\section{Experiment}

We have tested various methods presented so far. The methods used in our experiment
are summarized as follows.
\begin{itemize}
\item[\textbf{(sno)}]
$\mathrm{UNC}(\mathcal{R})$ if $\mathcal{R}$ is strongly non-overlapping.

\item[\textbf{($\omega$)}]
$\mathrm{UNC}(\mathcal{R})$ if $\mathcal{R}$ is non-$\omega$-overlapping.

\item[\textbf{(pcl)}]
$\mathrm{UNC}(\mathcal{R})$ if $\mathcal{R}^L$ is parallel-closed.

\item[\textbf{(scl)}]
$\mathrm{UNC}(\mathcal{R})$ if $\mathrm{UNC}(\mathcal{R})$ is right-linear and $\mathcal{R}^L$ is strongly closed.

\item[\textbf{(wd)}]
$\mathrm{UNC}(\mathcal{R})$ if $\mathcal{R}$ is non-duplicating and weight-decreasing joinable
by the condition of Theorem~\ref{thm:I}.

\item[\textbf{(sc)}]
UNC completion using strongly-closed critical pairs criterion for linear TRSs.

\item[\textbf{(dc)}]
UNC completion using development-closed critical pairs criterion for left-linear TRSs.

\item[\textbf{(rr)}]
$\mathrm{UNC}(\mathcal{R})$ if $\mathcal{R}$ is right-reducible.

\item[\textbf{(cp)}]
$\neg \mathrm{UNC}(\mathcal{R})$ 
by adhoc search of a counterexample for $\mathrm{UNC}(\mathcal{R})$.

\item[\textbf{(rev)}]
Rule reversing transformation, combined with other criteria above.

\end{itemize}
Here, we remark that \textbf{(sno)} is subsumed by \textbf{($\omega$)}
and just included for the reference.
For the implementation of non-$\omega$-overlapping condition, we need
unification over infinite terms; our implementation is based on the algorithm in  \cite{Jaffar:unifInf}.
The last one \textbf{(rev)} is used combined with the other methods.
For \textbf{(sc)} and \textbf{(dc)},
we employed an approximation of $\stackrel{*}{\to}$ by $\darrow$
in Step 3(d).
We employed a heuristics for \textbf{(rev)}
the first kind of transformation is tried only when the term length of $l$ is less than that of $r$.
For \textbf{(cp)}, we use an adhoc search based on rule reversing, critical pairs computation,
and rewriting.

We test on the 144 TRSs from the Cops (Confluence Problems) database\footnote{
Cops can be accessed from \url{http://cops.uibk.ac.at/}, which currently includes 438 TRSs.}
of which no confluence tool has proven confluence nor terminating.
The motivation of using such testbed is as follows:
If a confluent tool can prove CR, then UNC is obtained by confluent tools.
If $\mathcal{R}$ is terminating then 
$\mathrm{CR}(\mathcal{R})$ iff $\mathrm{UNC}(\mathcal{R})$,
and thus the result follows also from the result of confluence tools.
Assuming dedicated termination or confluence tools are used at first,
we haven't elaborated on sophisticated combination with confluence proofs in ACP.

In Table~\ref{table:test I}, we summarize the results. 
Out test is performed on a PC with 2.60GHz cpu with 4G of memory.
The column headings show the technique used.
The number of examples for which UNC is proved (disproved) 
successfully is shown in the row titled '\texttt{YES}' (resp.\ '\texttt{NO}').
In the columns below \textbf{(sc)} and \textbf{(dc)},
we put $l/n/m$ where each $l,n,m$ denotes the scores for the 
1-round  (2-rounds, 3-rounds) UNC completion.
The columns below '\textit{all}' show the numbers of examples
succeeded in any of the methods.

The columns below the row headed '\textit{with \textbf{(rev)}}' are 
the results for which methods are applied after the rule reversing transformation.
The columns below the row headed '\textit{both}'
show the numbers of examples succeeded by each technique, 
where the techniques are applied to both of the original TRSs and the TRSs obtained by the rule reversing transformation.

3 rounds UNC completions \textbf{(sc)}, \textbf{(dc)}
with rule reversing are most effective, but they also record most timeouts.
Simple methods \textbf{(rr)}, \textbf{(cp)} are also effective for not few examples.
There is only a small number of examples in the testbed for which 
weight-decreasing criterion or critical pairs criteria for conditional linearization work.
Rule reversing \textbf{(rev)} is only worth incorporated for UNC completions.
For other methods, the rule reversing make the methods less effective;
for methods \textbf{(sno)}, \textbf{($\omega$)},
\textbf{(pcl)}, \textbf{(scl)} and \textbf{(wd)},
this is because the rule reversing transformation generally increases the number of lhs of the rules.
In total, UNC of the 127 problems out of 144 problems have been solved by combining our techniques.
All the details of the experiment are found in
\url{http://www.nue.ie.niigata-u.ac.jp/tools/acp/experiments/ppdp18-sbm/}.

\begin{table*}
\begin{center}
\begin{minipage}{5.3cm}
\begin{tabular}{l|c|c|c}\hline
                           & \textsf{ACP}& \textsf{CSI} &\textsf{FORT} \\\hline
\texttt{YES}               &    66       &     41       &   38 \\
\texttt{NO}                &    60       &     43       &   34 \\
\texttt{YES}+\texttt{NO}   &   126       &     84       &   72 \\
time                       &   13m       &     25m      &   55s\\\hline
\end{tabular}
\end{minipage}
\begin{minipage}{3.8cm}
\begin{tikzpicture}[scale=.4,thick]
\draw (0,1) circle (1.8);
\draw (.866,-.5) circle (1.8);
\draw (-.866,-.5) circle (1.8);
\node at (0,0) {33};
\node at (0,-2.9) {67};
\node at (0,1.9) {22};
\node at (0,-1.3) {1};
\node at (1.645,-0.95) {0};
\node at (-1.645,-0.95) {0};
\node at (1.126,0.65) {4};
\node at (-1.195,0.65) {7};
\node at (0,3.4) {\textsf{ACP}};
\node at (3.8,-1.6) {\textsf{FORT}};
\node at (-3.5,-1.6) {\textsf{CSI}};
\end{tikzpicture}
\end{minipage}
\hspace{1cm}
\begin{minipage}{3.8cm}
\begin{tikzpicture}[scale=.4,thick]
\draw (0,1) circle (1.8);
\draw (.866,-.5) circle (1.8);
\draw (-.866,-.5) circle (1.8);
\node at (0,0) {26};
\node at (0,-2.9) {61};
\node at (0,1.9) {11};
\node at (0,-1.3) {1};
\node at (1.645,-0.95) {0};
\node at (-1.645,-0.95) {0};
\node at (1.126,0.65) {7};
\node at (-1.195,0.65) {16};
\node at (0,3.4) {\textsf{ACP}};
\node at (3.8,-1.6) {\textsf{FORT}};
\node at (-3.5,-1.6) {\textsf{CSI}};
\end{tikzpicture}
\end{minipage}
\end{center}
\caption{Comparison of UNC proofs (1)}
\label{table:test II}
\end{table*}

\begin{table*}
\begin{center}
\begin{minipage}{6cm}
\begin{tabular}{l|c|c|c|c}\hline
                           & \multicolumn{2}{c|}{58 ground} & \multicolumn{2}{c}{86 non-ground} \\\hline
                           & \textsf{ACP}& \textsf{CSI} & \textsf{ACP}& \textsf{CSI}\\\hline
\texttt{YES}               &    33       &     34       &   33        &     7       \\   
\texttt{NO}                &    23       &     24       &   37        &     19       \\  
\texttt{YES}+\texttt{NO}   &    56       &     58       &   70        &     26      \\  
time                       &    2.5m     &    70s       &   11.5m     &     25m \\\hline
\end{tabular}
\end{minipage}
\hspace{1cm}
\begin{minipage}{6cm}
\begin{tabular}{l|c|c|c|c}\hline
                           & \multicolumn{2}{c|}{72 LL-RG} & \multicolumn{2}{c}{72 non-LL-RG} \\\hline
                           & \textsf{ACP}& \textsf{FORT}  & \textsf{ACP}& \textsf{FORT}\\\hline
\texttt{YES}               &    37       &      38        &   29        &     --      \\   
\texttt{NO}                &    33       &      34        &   27        &     --       \\  
\texttt{YES}+\texttt{NO}   &    70       &      72        &   56        &     --      \\  
time                       &    2.8m     &      40s     &   11.4m       &     -- \\\hline
\end{tabular}
\end{minipage}
\end{center}
\caption{Comparison of UNC proofs (2)}
\label{table:test III}
\end{table*}

\section{Tool} 

The experiment in the previous section
reveals how presented methods for UNC (dis)proving should be combined---we 
have incorporated UNC (dis)proof methods
\textbf{($\omega$)},
\textbf{(pcl)},
\textbf{(scl)},
\textbf{(wd)},
\textbf{(rr)},
\textbf{(cp)},
\textbf{(rev+sc)}/3 and
\textbf{(rev+dc)}/3 into our confluence tool \textsf{ACP} \cite{ACP}.

\textsf{ACP} originally intends to (dis)prove confluence of TRSs;
we have extended it to also deal with (dis)proving UNC of TRSs.
Since \textsf{ACP} facilitates CR proof methods, 
it is easy to use confluence criteria other than strong-closedness and development-closedness;
thus, we add yet another UNC completion procedure in which 
confluence check is performed only to the final result of completion.

We have also incorporated modularity results:
we have incorporated extensions of Proposition \ref{prop:modular},
namely persistent decomposition \cite{AT97},
and layer-preserving decomposition \cite{AotoTandToyamaY:topdown}.
Using these decomposition methods, our tool first try to decompose 
the problem into smaller components if possible.


\textsf{ACP} is written in SML/NJ and provided as the heap image of SML.
The new version (ver.\ 0.62) 
is downloadable from\\
\hfil \url{http://www.nue.ie.niigata-u.ac.jp/tools/acp/}. \hfil\\
To run the UNC (dis)proving, it should be invoked like this:\\
\hspace*{1cm}\texttt{\$ sml @SMLload=acp.x86-linux -p unc}~\textit{filename}

Other tools that support UNC (dis)proving include \textsf{CSI} \cite{CSI2},
which is a powerful confluence prover supporting UNC proof
for non-$\omega$-overlapping TRSs and a decision procedure of UNC
for ground TRSs, 
and \textsf{FORT} \cite{FORT},
which implements decision procedure for first-order theory of
left-linear right-ground TRSs based on tree automata.
Our new methods are also effective for TRSs outside 
the class of non-$\omega$-overlapping TRSs
and that of left-linear right-ground TRSs.

A comparison of our tool and these tools 
(\textsf{CSI} ver.~1.1 and \textsf{FORT} ver.~1.0)
is given in Table~\ref{table:test II} (a).
The diagram on the center (right) in Table 2 shows the distribution of problems 
for which some tool can show UNC (resp. non-UNC).
There are 22 problems for which UNC has been newly proved automatically,
and 11 problems for which UNC have been newly disproved automatically.
Since most of success of \textsf{CSI} (\textsf{FORT})
is due to the decision procedure for ground TRSs (left-linear right-ground TRSs),
the size of the problem sets of ground TRS vs.\ non-ground TRSs
(left-linear right-ground TRSs vs.\ non-left-linear or non-right-ground TRSs)
highly affect the result.
In Table~\ref{table:test III}(b),  we present a comparison of
\textsf{ACP} and \textsf{CSI} distinguishing case of ground TRSs and non-ground TRSs,
and that of \textsf{ACP} and \textsf{FORT} distinguishing case of left-linear right-ground TRSs
and other TRSs.
Our methods work for many of left-linear right-ground TRSs, but takes much longer time
than decision procedures in \textsf{CSI} or \textsf{FORT}.

\section{Conclusion}

In this paper, we have studied automated methods for (dis)proving UNC of TRSs.
We have presented some new methods for (dis)proving UNC of TRSs.
Presented methods, 
except for 
the decidability results (Propositions~\ref{prop:dec LLRG} and \ref{prop:dec shallow}),
have been implemented over our confluence tool \textsf{ACP}.
Our tool is capable of UNC (dis)proofs for TRSs outside 
the class of non-$\omega$-overlapping TRSs and that of left-linear right-ground TRSs,
for which class UNC dis(proof) had been already implemented by tools
\textsf{CSI} and \textsf{FORT}, respectively. 

We have not yet incorporated 
the decidability results (Propositions~\ref{prop:dec LLRG} and \ref{prop:dec shallow}).
Currently, our tool lacks 
a sophisticated infrastructure for implementing efficient decision procedures.
Incorporating these methods to our tool remains as our future work.
It is shown in \cite{condlin} that $CL_\mathrm{sp}$,
the S,K,I-rules for the combinatory logic supplemented with surjective pairing, 
has UNC. Our tool, however, can not handle this example;
this is theoretically so, even with the help of Propositions~\ref{prop:dec LLRG} and \ref{prop:dec shallow}.
The argument used in \cite{condlin} for showing UNC of $CL_\mathrm{sp}$ seems hardly automatable.
Thus, more powerful methods to prove UNC automatically should be investigated.
Lastly, another future plan is to extend our tools to deal with NFP and UNR,
and conditional rewriting as well.

%

%


\begin{thebibliography}{DHLT90}

\bibitem[AT96]{AotoTandToyamaY:topdown}
T.~Aoto and Y.~Toyama.
\newblock Top-down labelling and modularity of term rewriting systems.
\newblock Research Report IS-RR-96-0023F, School of Information Science, JAIST,
  1996.

\bibitem[AT97]{AT97}
T.~Aoto and Y.~Toyama.
\newblock On composable properties of term rewriting systems.
\newblock In {\em Proc.\ of 6th ALP and 3rd HOA}, volume 1298 of {\em LNCS},
  pages 114--128. Springer-Verlag, 1997.

\bibitem[AYT09]{ACP}
T.~Aoto, Y.~Yoshida, and Y.~Toyama.
\newblock Proving confluence of term rewriting systems automatically.
\newblock In {\em Proc.\ of 20th RTA}, volume 5595 of {\em LNCS}, pages
  93--102. Springer-Verlag, 2009.

\bibitem[BK86]{BergstralJAandKlopJW:crs}
J.~A. Bergstra and J.~W. Klop.
\newblock Conditional rewrite rules: confluence and termination.
\newblock {\em Journal of Computer and System Sciences}, 32:323--362, 1986.

\bibitem[BN98]{BaaderFandNipkowT:TR}
F.~Baader and T.~Nipkow.
\newblock {\em Term Rewriting and All That}.
\newblock Cambridge University Press, 1998.

\bibitem[Che81]{Chew}
P.~Chew.
\newblock Unique normal forms in term rewriting systems with repeated
  variables.
\newblock In {\em Proc.\ of 13th STOC}, pages 7--18, 1981.

\bibitem[DHLT90]{DecCR}
M.~Dauchet, T.~Heuillard, P.~Lescanne, and S.~Tison.
\newblock Decidability of the confluence of finite ground term rewrite systems
  and of other related term rewrite systems.
\newblock {\em Information and Computation}, 88:187--201, 1990.

\bibitem[dV99]{condlin}
R.C. de~Vrijer.
\newblock Conditional linearization.
\newblock {\em Indagationes Mathematicae}, 10(1):145--159, 1999.

\bibitem[Fel16]{UNground}
B.~Felgenhauer.
\newblock Efficiently deciding uniqueness of normal forms and unique
  normalization for ground trss.
\newblock In {\em Proc.\ of 5th IWC}, pages 16--20, 2016.

\bibitem[Gra96]{Gra96caap}
B.~Gramlich.
\newblock Confluence without termination via parallel critical pairs.
\newblock In {\em Proc.\ of 21st CAAP}, volume 1059 of {\em LNCS}, pages
  211--225. Springer-Verlag, 1996.

\bibitem[HK12]{Saigawa}
N.~Hirokawa and D.~Klein.
\newblock Saigawa: a confluence tool.
\newblock In {\em Proc.\ of 1st IWC}, page~49, 2012.

\bibitem[Hue80]{Hue80}
G.~Huet.
\newblock Confluent reductions: abstract properties and applications to term
  rewriting systems.
\newblock {\em Journal of the ACM}, 27(4):797--821, 1980.

\bibitem[Jaf84]{Jaffar:unifInf}
J.~Jaffar.
\newblock Efficient unification over infinite trees.
\newblock {\em New Generation Computing}, 2:207--219, 1984.

\bibitem[KdV90]{KV90}
J.W. Klop and R.~de~Vrijer.
\newblock Extended term rewrite systems.
\newblock In {\em Proc.\ of 2nd CTRS}, volume 516 of {\em LNCS}, pages 26--50.
  Springer-Verlag, 1990.

\bibitem[Klo80]{Klo80}
J.W. Klop.
\newblock {\em Combinatory Reduction Systems}, volume 127 of {\em Mathematical
  Centre Tracts}.
\newblock CWI, Amsterdam, Holland, 1980.

\bibitem[KS16]{UNomega}
S.~Kahrs and C.~Smith.
\newblock Non-$\omega$-overlapping {TRS}s are {UN}.
\newblock In {\em Proc.\ of 1st FSCD}, volume~52 of {\em LIPIcs}, pages
  22:1--17. Schloss Dagstuhl, 2016.

\bibitem[Mid90]{MiddeldorpA:phd}
A.~Middeldorp.
\newblock {\em Modular Properties of Term Rewriting Systems}.
\newblock PhD thesis, Vrije Universiteit, Amsterdam, 1990.

\bibitem[MO01]{ManoOgawaTCS}
K.~Mano and M.~Ogawa.
\newblock Unique normal form property of compatible term rewriting systems: a
  new proof of chew's theorem.
\newblock {\em Theoretical Computer Science}, 258:169--208, 2001.

\bibitem[MOJ06]{undecflat}
I.~Mitsuhashi, M.~Oyamaguch, and F.~Jacquemard.
\newblock The confluence problem for flat trss.
\newblock In {\em Proc.\ of 8th AISC}, volume 4120 of {\em LNAI}, pages 68--81.
  Springer-Verlag, 2006.

\bibitem[NFM17]{CSI2}
J.~Nagele, B.~Felgenhauer, and A.~Middeldorp.
\newblock {CSI}: new evidence - a progress report.
\newblock In {\em Proc.\ of 26th CADE}, volume 10395 of {\em LNAI}, pages
  385--397. Springer-Verlag, 2017.

\bibitem[O'D77]{ODonnell}
M.J. O'Donnell.
\newblock {\em Computing in systems described by equations}, volume~58 of {\em
  LNCS}.
\newblock Springer-Verlag, 1977.

\bibitem[Oku98]{Oku98}
S.~Okui.
\newblock Simultaneous critical pairs and {C}hurch-{R}osser property.
\newblock In {\em Proc.\ of 9th RTA}, volume 1379 of {\em LNCS}, pages 2--16.
  Springer-Verlag, 1998.

\bibitem[RM16]{FORT}
F.~Rapp and A.~Middeldorp.
\newblock Automating the first-order theory of rewriting for left-linear
  right-ground rewrite systems.
\newblock In {\em Proc.\ of 1st FSCD}, volume~52 of {\em LIPIcs}, pages
  36:1--17. Schloss Dagstuhl, 2016.

\bibitem[RMV17]{UNCshallow}
N.~R. Radcliffe, L.~F.~T. Moreas, and R.~M. Verma.
\newblock Uniqueness of normal forms for shallow term rewrite systems.
\newblock {\em ACM Transactions on Computational Logic}, 18(2):17:1--17:20,
  2017.

\bibitem[TO01]{WeightDec}
Y.~Toyama and M.~Oyamaguchi.
\newblock Conditional linearization of non-duplicating term rewriting systems.
\newblock {\em IEICE Transactions on Information and Systems},
  E84-D(4):439--447, 2001.

\bibitem[Toy88]{Toyama88}
Y.~Toyama.
\newblock Commutativity of term rewriting systems.
\newblock In K.~Fuchi and L.~Kott, editors, {\em Programming of Future
  Generation Computers II}, pages 393--407. North-Holland, 1988.

\bibitem[Toy05]{ToyRTA}
Y.~Toyama.
\newblock Confluent term rewriting systems (invited talk).
\newblock In {\em Proc.\ of 16th RTA}, volume 3467 of {\em LNCS}, page~1.
  Springer-Verlag, 2005.
\newblock Slides are available from
  \url{http://www.nue.ie.niigata-u.ac.jp/toyama/user/toyama/slides/toyama-RTA05.pdf}.

\bibitem[vO97]{Dev}
V.~van Oostrom.
\newblock Developing developments.
\newblock {\em Theoretical Computer Science}, 175(1):159--181, 1997.

\end{thebibliography}

\appendix

\section{Omitted Proofs}
\label{sec:omitted proofs}







We first prepare two lemmas to present
a proof of Theorem \ref{thm:parallel closed}.



\begin{lemma}
\label{lem:cp lemma I}
Let $\mathcal{R}$ be a semi-equational CTRS
and $l \to r \Leftarrow \Gamma \in \mathcal{R}$ be left-linear.
Suppose $s  \mathrel{_P{\prarrow}} l\theta \to_{\epsilon, l \to r\Leftarrow \Gamma} r\theta$,
and any redex occurrence of $l\theta \to_P s$ is contained in 
a subterm occurrence of $\theta(x)$ in $l\theta$ for some $x \in \mathcal{V}(l)$.
Then there exists $t$ such that 
$s \to_{\epsilon, l \to r \Leftarrow \Gamma} t \prarrow r\theta$.
\end{lemma}

\begin{proof}
Let $P = \{ p_1,\ldots,p_k \}$,
and, for each $1 \le i \le k$, 
let $\alpha_i$ be the subterm occurrence in $l\theta$ at $p_i$ 
and $\beta_i$ be the subterm occurrence in $s$ at $p_i$.
For each $x \in \mathcal{V}(l)$,
let $\theta(x) = C_x[\alpha_{i_1},\ldots,\alpha_{i_m}]$ with all 
$\alpha_1,\ldots,\alpha_k$ in $\theta(x)$ displayed.
Take a substitution $\theta'$ such as $\theta'(x) = C_x[\beta_{i_1},\ldots,\beta_{i_m}]$.
Then, we have $s = l\theta'$ by linearity of $l$,
and moreover, $\theta'(y) \stackrel{*}{\leftrightarrow} \theta(y)$ for all $y \in \mathcal{V}$
by definition.
From the latter and $\vdash_\mathcal{R} \Gamma\theta$,
we obtain $\vdash_\mathcal{R} \Gamma\theta'$.
Thus, $s = l\theta' \to_\epsilon r\theta'$.
Let $r = C'[x_1,\ldots,x_n]$ with all variable occurrences in $r$ displayed.
Then $r\theta = C'[x_1\theta,\ldots,x_n\theta]
\plarrow C'[x_1\theta',\ldots,x_n\theta'] = r\theta'$.
Thus, 
$s \to_{\epsilon, l \to r \Leftarrow \Gamma} r\theta' \prarrow r\theta$
and the claim is obtained.
\end{proof}

\begin{lemma}
\label{lem:cp lemma II}
Let $\mathcal{R}$ be a semi-equational CTRS and 
$l_1 \to r_1 \Leftarrow \Gamma_1 \in \mathcal{R}$.
Suppose $s  \mathrel{_{p}{\gets}} l_1\theta \to_{\epsilon, l_1 \to r_1\Leftarrow \Gamma_1} r_1\theta$,
and the redex occurrence of $l_1\theta \to_p s$ is not contained in 
any subterm occurrence of $\theta(x)$ ($x \in \mathcal{V}(l_1)$) in $l_1\theta$.
Then $s  \gets l_1\theta \to r_1\theta$
is an instance of a conditional critical pair
$\Sigma \Rightarrow \langle v, w \rangle$ and substitution $\sigma$,
i.e.\ there exists some substitution $\sigma$
such that
$s = v\sigma$, $r_1\theta = w\sigma$ and $\vdash_\mathcal{R} \Sigma\sigma$.
\end{lemma}

\begin{proof}
Let $s  \mathrel{_{p,l_2 \to r_2\Leftarrow \Gamma_2}{\gets}} l_1\theta$.
W.l.o.g.\ assume $\mathcal{V}(l_1) \cap \mathcal{V}(l_2) = \emptyset$.
Then we can let $l_1\theta = l_1[l_2]_p \theta$ and $s = l_1\theta[r_2\theta]_p$.
By the condition $p \in \mathrm{Pos}_\mathcal{F}(l_1)$,
and hence $l_1\theta|_p = l_1|_p\theta = l_2 \theta$,
and thus $l_1|_p$ and $l_2$ is unifiable.
Hence, there exists a conditional critical pair
$\Gamma_1\rho,\Gamma\rho \Rightarrow \langle l_1[r_2]_p\rho, r_1 \rho \rangle$ of $\mathcal{R}$,
where $\rho$ is an mgu of $l_1|_p$ and $l_2$.
Furthermore, by the definition of mgu, there exists a substitution $\sigma$
such that $\sigma \circ \rho = \theta$.
Then we have $s = l_1[r_2]_p\theta = l_1[r_2]_p\rho\sigma = v\sigma$,
$r_1\theta = r_1\rho\sigma = w\sigma$, and
$\Gamma_1\theta \cup \Gamma_2\theta
= (\Gamma_1 \rho \cup \Gamma_2\rho) \sigma
= \Sigma\sigma$.
Thus, from $\vdash_\mathcal{V} \Gamma_1\theta, \Gamma_2\theta$,
it follows $\vdash_\mathcal{V} \Sigma\sigma$.
\end{proof}

\begin{proof}[Proof of Theorem \ref{thm:parallel closed}]
We show the claim
$t \plarrow t_1$ and
$t \plarrow t_2$
imply 
$t_1 \plarrow^* t_3$ and $t_2 \plarrow  t_3$ for some $t_3$.
In fact, the proof is almost same as that of the criteria for TRSs.
The only essential difference is 
captured by Lemmas~\ref{lem:cp lemma I} and \ref{lem:cp lemma II}.
For such parallel peak,
let $t \plarrow_{P_1} t_1$ with $P_1 = \{ p_{11},\ldots,p_{1m} \}$
and $t \plarrow_{P_2} t_2$ with $P_2 = \{ p_{21},\ldots,p_{2n} \}$.
We set subterm occurrences $\alpha_i = t|_{p_{1i}}$ for $i = 1,\ldots,m$
and $\beta_j = t|_{p_{2j}}$ for $j = 1,\ldots,n$.
Let $t = C_1[\alpha_1,\ldots,\alpha_{m}]_{p_{11},\ldots,p_{1m}}
       = C_2[\beta_1,\ldots,\beta_{n}]_{p_{21},\ldots,p_{2n}}$.
Let $t|_{p_k} = l_k\sigma_k$ with $l_k \to r_k \Leftarrow \Gamma_k \in \mathcal{R}$.
Then, we have $t_1 = C_1[r_{11}\sigma_{11},\ldots,r_{1m}\sigma_{1m}]$
$t_2 = C_2[r_{21}\sigma_{21},\ldots,r_{2n}\sigma_{2n}]$,
and $\vdash_\mathcal{R} \Gamma_k\sigma_k$ for all $p_k \in P_1 \cup P_2$.
Let 
\[
\begin{array}{lcl}
\textit{Red}_\textit{in}(t_1 \prarrow t \plarrow t_2)
    &=& \{ \alpha_i  \mid \exists j. \alpha_i \subset \beta_j  \}
   \uplus \{ \beta_j \mid \exists i. \beta_j \subseteq \alpha_i   \}\\
\textit{Red}_\textit{out}(t_1 \prarrow t \plarrow t_2)
    &=& \{ \alpha_i  \mid \forall j. \alpha_i \not\subset \beta_j  \}
   \uplus \{ \beta_j \mid \forall i. \beta_j \not\subseteq \alpha_i  \}
\end{array}
\]
Let us denote by $\vert t \vert$ the size of a term $t$.
Let $\vert M \vert = \sum_{t \in M} \vert t \vert$.
The proof of the claim is by induction on $\vert \textit{Red}_\textit{in}(t_1 \prarrow t \plarrow t_2) \vert$.

\begin{itemize}
\item Case $\vert I \vert = 0$.
Then for any $p_{k_1},p_{k_2} \in P_1 \cup P_2$, $k_1 \neq k_2$ implies $p_{k_1} \parallel p_{k_2}$.
For notational simplicity, we only consider the case
$t = C[\alpha_1,\ldots,\alpha_m,\beta_1,\ldots,\beta_n]$,
$t_1 = C[\alpha_1',\ldots,\alpha_m',\beta_1,\ldots,\beta_n]$,
$t_2 = C[\alpha_1,\ldots,\alpha_m,\beta_1',\ldots,\beta_n']$,
with 
$\alpha_{i} \to \alpha_{i}'$ ($1 \le i \le m$)
and 
$\beta_{j} \to \beta_{j}'$ ($1 \le j \le n$).
Let $t_3 = C[\alpha_1',\ldots,\alpha_m',\beta_1',\ldots,\beta_n']$,
Then 
$t_1 \plarrow_{P_2} t_3$ and $t_2 \plarrow_{P_1}  t_3$.

\item Case $\vert I \vert > 0$.

Let $\gamma_1,\ldots, \gamma_{h}$ be subterm occurrences of the term 
$t$ contained in
$\textit{Red}_\textit{out}(t_1 \prarrow t \plarrow t_2)$.

Then we can write $t = C'[\gamma_1,\ldots,\gamma_h]$,
$t_1 = C'[\gamma_{11},\ldots,\gamma_{1h}]$,
$t_2 = C'[\gamma_{21},\ldots,\gamma_{2h}]$,
where, for each $1 \le k  \le h$, 
$\gamma_k \plarrow \gamma_{1k}$ and $\gamma_1 \plarrow \gamma_{2k}$
with one of them being a root step.
It is sufficient to show there are $\gamma_1',\ldots,\gamma_h'$ 
such that $\gamma_{1k} \plarrow^* \gamma_{k}'$ and $\gamma_{2k} \plarrow \gamma_{k}'$ for each $1 \le k \le h$.

Suppose $1 \le k  \le h$.
\begin{itemize}
\item
Let us consider the case $\gamma_k \plarrow_{\{ \epsilon \}} \gamma_{1k}$ and $\gamma_k \plarrow_P  \gamma_{2k}$.
Then there exist $l \to r \Leftarrow \Gamma \in \mathcal{R}$ and $\sigma$
such that $\gamma_k = l\sigma$ and $\gamma_{1k} = r\sigma$ and  $\vdash_\mathcal{R} \Gamma\sigma$.
Let $\gamma_{k} = \hat{C}[\hat{\gamma}_1,\ldots,\hat{\gamma}_g]$
where the subterm occurrences $\hat{\gamma}_1,\ldots,\hat{\gamma}_g$ 
are at the respective positions in $P$.
Then we can let $\gamma_{2k} = \hat{C}[\hat{\gamma}_1',\ldots,\hat{\gamma}_g']$
with $\hat{\gamma}_i \to \hat{\gamma}_i'$ for each $1 \le i \le g$.

First, consider the case that that for each $\hat{\gamma}_i$,
there exists $x \in \mathcal{V}(l)$ such that $\hat{\gamma}_i$ is contained in some $\sigma(x)$.
Then, by Lemma \ref{lem:cp lemma I},  
$\gamma_{2k} \to \circ \prarrow \gamma_{1k}$.

Otherwise, there exists some $1 \le i \le g$ 
such that $\hat{\gamma}_i$ is contained in $\sigma(x)$ for no $x \in \mathcal{V}(l)$.
Let $p$ be the position of $\hat{\gamma}_i$ in $\gamma_k$.
Then we have
$\gamma_k \to_p \gamma_k[\hat{\gamma}_i']_p \plarrow_{P \setminus \{ p \}} \gamma_{2k}$.
Then, by Lemma~\ref{lem:cp lemma II},
$\langle \gamma_k[\hat{\gamma}'_i], \gamma_{1k}  \rangle$ is
an instance of some CCP $\Gamma \Rightarrow \langle u,v \rangle$,
i.e.\
there exists some $\theta$ such that
$\gamma_k[\hat{\gamma}'_i] = u\theta$,
$\gamma_{1k} = v\theta$ and 
$\vdash_\mathcal{R} \Gamma\theta$.
We distinguish two cases.
\begin{itemize}
\item Case $p = \epsilon$.
Then we have $P = \{ \epsilon \}$
and $\gamma_k[\hat{\gamma}_i']_p = \gamma_{2k}$.
Furthermore, $\Gamma \Rightarrow \langle u,v \rangle$ is an overlay critical pair,
and hence, we have $\Gamma \vdash_\mathcal{R} u \plarrow \circ \prarrow^*  v$ by the parallel-closed assumption.
Thus, 
$u\theta \plarrow \circ \prarrow^*  v\theta$ follows from $\vdash_\mathcal{R} \Gamma\theta$
by Definition.
Hence we have 
$\gamma_{1k} = v\theta \plarrow^*  \circ \prarrow u\theta  = \gamma_k[\hat{\gamma}'_i]  = \gamma_{2k}$.

\item Case $p \neq \epsilon$.
Then, $\Gamma \Rightarrow \langle u,v \rangle$ is an inner-outer critical pair.
Hence, we have $\Gamma \vdash_\mathcal{R} u \plarrow v$ by the parallel-closed assumption.
Thus, 
$u\theta \plarrow v\theta$ follows from $\vdash_\mathcal{R} \Gamma\theta$
by Definition.
Hence we have 
$\gamma_{1k} = v\theta \prarrow u\theta  = \gamma_k[\hat{\gamma}'_i]  \plarrow_{P \setminus \{ p \}} \gamma_{2k}$.
Now, 
$\textit{Red}_\textit{in}(t_1 \prarrow \circ \plarrow t_2)$
contains $\hat{\gamma}_1,\ldots,\hat{\gamma}_d$.
On the other hand, 
$\textit{Red}_\textit{in}(\gamma_{1k} \prarrow \gamma_k[\hat{\gamma}'_i] \plarrow \gamma_{2k})$
contains only subterm occurrences of 
$\hat{\gamma}_1,\ldots,\hat{\gamma}_{p-1},\hat{\gamma}_{p+1},\ldots,\hat{\gamma}_d$.
Thus, we have 
$\vert \textit{Red}_\textit{in}(\gamma_{1k} \prarrow \gamma_k[\hat{\gamma}'_i]\plarrow \gamma_{2k}) \vert
< \vert \textit{Red}_\textit{in}(t_1 \prarrow \circ \plarrow t_2) \vert$.
Thus, by induction hypothesis,
$\gamma_{1k} \plarrow^* \circ \prarrow \gamma_{2k}$.
\end{itemize}

\item
The case $\gamma_k \plarrow_P \gamma_{1k}$ and $\gamma_k \plarrow_{\{ \epsilon \}} \gamma_{2k}$.
This case is proved analogously to the previous case.
\end{itemize}
\end{itemize}
\end{proof}



\begin{proof}[Proof of Theorem~\ref{thm:decidability of WDJ}]
We here supplement the proof of Theorem~\ref{thm:decidability of WDJ}.
For (c), take $S_{l \to r\Leftarrow c}(s,t) = \{ \sigma \mid C[l\sigma] = s,  C[l\sigma] = t  \}$
and then
$\mathrm{RED}_1(\Gamma, s, t) =
\bigcup_{l \to r \Leftarrow c \in \mathcal{R}} \{ \Sigma \mid \langle \Sigma,\mathsf{rhs}(c\sigma)\rangle 
\in  \mathrm{SIM}_0(\Gamma, \mathsf{lhs}(c\sigma)),
\sigma \in S_{l \to r \Leftarrow c}(s,t) \}$,
where $\mathsf{lhs}(u_1\sigma\approx v_1\sigma,\ldots,u_n\sigma \approx v_n\sigma) = \langle u_1\sigma,\ldots,u_n\sigma \rangle$
and $\mathsf{rhs}(u_1\sigma\approx v_1\sigma,\ldots,u_n\sigma \approx v_n\sigma) = \langle v_1\sigma,\ldots,v_n\sigma \rangle$.
For (d), take $A = \bigcup_{(\Psi,s') \in \mathrm{SIM}_0(\Gamma, s)} \{ \langle\Gamma',s',t'\rangle \mid (\Gamma',t') \in \mathrm{SIM}_0(\Psi, t) \}$
and $\bigcup \{ \mathrm{RED}_1(\Gamma',s',t')  \mid \langle\Gamma',s',t'\rangle \in A \}$.
For (g), as ${\sim}_1 =  {\sim}_0 \circ {\leftrightarrow}_1 \circ {\sim}_0$,
take $\mathrm{SRS}_{010}(\Gamma, s, t) \cup  \mathrm{SRS}_{010}(\Gamma, t, s)$.

Now, 
the condition (i) is equivalent to $\langle \Sigma, t \rangle \in \mathrm{SIM}_0(\Gamma,s)$ for some $\Sigma$
or $\mathrm{SIM}_1(\Gamma,s,t) \neq \emptyset$.
The condition (ii) is equivalent to $\mathrm{RED}_2(\Gamma,s,t) \cup  \mathsf{RED}_2(\Gamma,t,s) \neq \emptyset$.
The first part of condition (iii) is equivalent to 
(a) $\Gamma \Vdash_{\mathcal{R}} s  \to_2 \circ \sim_0 t$ or 
(b) $\Gamma \Vdash_{\mathcal{R}} s  \to_1 \circ \sim_1 t$ or 
(c) $\Gamma \Vdash_{\mathcal{R}} s  \to_1 \circ \sim_0 t$.
(a,c) is equivalent to $ \mathrm{RED}_1(\Sigma,s,t') \cup \mathrm{RED}_2(\Sigma,s,t') \neq \emptyset$ 
for some $\langle \Sigma, t' \rangle  \in \mathrm{SIM}_0(\Gamma, t)$.
(b) is equivalent to
$\mathrm{SIM}_1(\Sigma,s',t) \neq \emptyset$ 
for some $\langle \Sigma, s' \rangle  \in \mathrm{RED}_1(\Gamma,s)$.
The second part is similar.
\end{proof}

\begin{proof}[Proof of Lemma \ref{lem:approximation II}]
(i) Suppose $s \stackrel{*}{\leftrightarrow}_\mathcal{R} t$
and $s,t \in \mathrm{NF}(\mathcal{R})$.
Then $s \stackrel{*}{\to}_\mathcal{R} w \mathrel{{_\mathcal{R}{\stackrel{*}{\gets}}}} t$ for some $w$ by $\mathrm{CR}(\mathcal{R})$.
But by $s,t \in \mathrm{NF}(\mathcal{S})$, we obtain $s = w = t$.
(ii) From ${\to}_\mathcal{R} \subseteq {\to}_{\mathcal{S}}$,
we have $\mathrm{NF}(\mathcal{S}) \subseteq \mathrm{NF}(\mathcal{R})$,
and thus $s,t \in \mathrm{NF}(\mathcal{R})$.
From ${\to}_{\mathcal{S}} \subseteq {\stackrel{*}{\leftrightarrow}}_\mathcal{R}$,
$s \stackrel{*}{\leftrightarrow}_\mathcal{R} t$.
\end{proof}

\begin{proof}[Proof of Lemma~\ref{lem:addition}]
($\Rightarrow$)
Suppose $s  \stackrel{*}{\leftrightarrow}_{\mathcal{R} \cup \{ l \to r \}} t$
with $s,t \in \mathrm{NF}(\mathcal{R} \cup \{ l \to r \})$.
Then from $l \stackrel{*}{\leftrightarrow}_\mathcal{R} r$,
we have $s  \stackrel{*}{\leftrightarrow}_{\mathcal{R}} t$.
Furthermore, by $l \notin \mathrm{NF}(\mathcal{R})$,
$\mathrm{NF}(\mathcal{R}) = \mathrm{NF}(\mathcal{R} \cup \{ l \to r \})$.
Thus, $s  \stackrel{*}{\leftrightarrow}_{\mathcal{R}} t$ and $s,t  \in \mathrm{NF}(\mathcal{R})$.
Hence $s = t$ by $\mathrm{UNC}(\mathcal{R})$.
($\Leftarrow$)
Suppose $s  \stackrel{*}{\leftrightarrow}_{\mathcal{R}} t$
with $s,t \in \mathrm{NF}(\mathcal{R})$.
Then, by $\mathcal{R} \subseteq \mathcal{R} \cup \{ l \to r \}$,
we have $s  \stackrel{*}{\leftrightarrow}_{\mathcal{R} \cup \{ l \to r \}} t$.
Furthermore, by $l \notin \mathrm{NF}(\mathcal{R})$,
$\mathrm{NF}(\mathcal{R}) = \mathrm{NF}(\mathcal{R} \cup \{ l \to r \})$.
Thus, 
Suppose $s  \stackrel{*}{\leftrightarrow}_{\mathcal{R} \cup \{ l \to r \}} t$
with $s,t \in \mathrm{NF}(\mathcal{R} \cup \{ l \to r \})$.
Hence,  $s = t$ by $\mathrm{UNC}(\mathcal{R}\cup \{ l \to r \})$.
\end{proof}

\begin{proof}[Proof of Lemma~\ref{lem:disproving by eliminating variable}]
Suppose
$s \stackrel{*}{\leftrightarrow}_\mathcal{R} t \in \mathrm{NF}(\mathcal{R})$
and $x \in \mathcal{V}(t) \setminus \mathcal{V}(s)$.
Take a fresh variable $y$  and let $t' = t\{ x:= y \}$.
Clearly, from $t \in \mathrm{NF}(\mathcal{R})$ 
we have $t' \in \mathrm{NF}(\mathcal{R})$.
By $t' \stackrel{*}{\leftrightarrow}_\mathcal{R} s \stackrel{*}{\leftrightarrow}_\mathcal{R} t$,
we obtain the claim.
\end{proof}

\begin{proof}[Proof of Theorem~\ref{thm:correctness of UNC procedure}]
By Lemma~\ref{lem:addition}, each round Step 1--4 keeps whether $\mathrm{UNC}(\mathcal{R})$ or not.
By the assumption on $\varphi,\Psi$, if UNC is returned in Step 2,
then $\mathrm{UNC}(\mathcal{R})$ holds.
In Step 3a, $u, v$ are convertible distinct normal forms and
hence $\neg \mathrm{UNC}(\mathcal{R})$ holds.
In Step 3b/c, $\neg \mathrm{UNC}(\mathcal{R})$ holds
by Lemma~\ref{lem:disproving by eliminating variable}.
\end{proof}

\begin{proof}[Proof of Theorem~\ref{thm:UNC of rule reversing transformation}]
It suffices to show $\mathcal{R} \leadsto \mathcal{R}'$ implies 
$\mathrm{UNC}(\mathcal{R})$ iff $\mathrm{UNC}(\mathcal{R}')$.
It is easy to see that both conditions ensure
that $\mathrm{NF}(\mathcal{R}) = \mathrm{NF}(\mathcal{R}')$
and ${\stackrel{=}{\leftrightarrow}_\mathcal{R}} = {\stackrel{=}{\leftrightarrow}_{\mathcal{R}'}}$.
From the latter, ${\stackrel{*}{\leftrightarrow}_\mathcal{R}} = {\stackrel{*}{\leftrightarrow}_{\mathcal{R}'}}$.
Thus the claim follows.
\end{proof}

\begin{proof}[Proof of Theorem~\ref{thm:unc by rhs}]
Suppose $s \stackrel{*}{\leftrightarrow}_\mathcal{R} t$, $s,t \in \mathrm{NF}(\mathcal{R})$ and $s \neq t$.
Then from $s \neq t$, we have $s \stackrel{+}{\leftrightarrow}_\mathcal{R} t$,
and thus $s \leftrightarrow_\mathcal{R} s' \stackrel{*}{\leftrightarrow}_\mathcal{R} t$ for some $s'$.
If $s \to_\mathcal{R}  s'$ then this contradicts $s \in \mathrm{NF}(\mathcal{R})$.
If $s' \to_\mathcal{R} s$ then $s' = C[l\theta]$ 
and $s = C[r\theta]$ for some $l \to r \in \mathcal{R}$,
and hence from $r \notin \mathsf{NF}(\mathcal{R})$
we know $s \notin \mathsf{NF}(\mathcal{R})$.
This is again a contradiction.
\end{proof}

\section{Comparison to our Definition \ref{def:derivation of sim}
and Definition 9 of \cite{WeightDec} and a proof of Theorem \ref{thm:I}}
\label{sec:Comparison to WeightDec}

\newcommand{\simto}{\mathrel{{\sim}{\triangleright}}}

The following definition is obtained by
adding the rule \textbf{(refl)} to the Definition 9 of \cite{WeightDec}.

\begin{definition}
\label{def:new sim}
Let $\mathcal{R}$ be a non-duplicating LR-separated CTRS.
Let $\Gamma$ be a multiset of equations $t' \approx s'$
and a fresh constant $\bullet$.
Then relations $t \underset{\Gamma}{\sim} s$ and $t \underset{\Gamma}{\simto} s$ on terms 
are inductively defined as follows:
\begin{description}
\item[\textbf{(asp)}] $t \underset{\{ t \approx s \}}{\sim} s$.
\item[\textbf{(refl)}] $t \underset{\{   \}}{\sim} t$.
\item[\textbf{(sym)}] If $t \underset{\Gamma}{\sim} s$ then $s \underset{\Gamma}{\sim} t$.
\item[\textbf{(trans)}] If $t \underset{\Gamma}{\sim} r$ and $r \underset{\Gamma'}{\sim} s$ then $t \underset{\Gamma \sqcup \Gamma'}{\sim} s$.
\item[\textbf{(cntxt)}] If $t \underset{\Gamma}{\sim} s$ then $C[t] \underset{\Gamma}{\sim} C[s]$.
\item[\textbf{(rule)}] If $l \to r \Leftarrow x_1 \approx y_1,\ldots,x_n \approx y_n \in \mathcal{R}$ and
$x_1\theta \underset{\Gamma_i}{\sim} y_i\theta$  $(i = 1,\ldots,n)$ then
$C[l\theta] \underset{\Gamma}{\simto} C[r\theta]$ 
where $\Gamma = \Gamma_1 \sqcup \cdots \sqcup \Gamma_n$.
\item[\textbf{(bullet)}] If $t \underset{\Gamma}{\simto} s$ then $t \underset{\Gamma \sqcup \{ \bullet \}}{\sim} s$.
\end{description}
\end{definition}

Note 
$t \underset{\Gamma}{\sim} s$ in the sense of Definition 9 of \cite{WeightDec}
implies 
$t \underset{\Gamma}{\sim} s$ in the sense of Definition \ref{def:new sim}.
On the other hand,
$t \underset{\Gamma}{\sim} s$ in the sense of Definition \ref{def:new sim}
uses \textbf{(refl)} rule in the derivation,
then 
$t \underset{\Gamma}{\sim} s$ in the sense of Definition 9 of \cite{WeightDec} does not hold.

Now, Lemma 3 of \cite{WeightDec} also follows for our Definition
of $\underset{\Gamma}{\sim}$ and $\underset{\Gamma}{\simto}$,
since the claim holds for the \textbf{(refl)} case trivially.

\begin{lemma}[Lemma 3 of \cite{WeightDec}, generalized]
Let $\Gamma = \{ p_1 \approx q_1, \ldots, p_m \approx q_m, \bullet,\ldots,\bullet \}$
be a multiset in which $\bullet$ occurs $k$ times ($k ≥ 0$), and
let $\mathcal{P}_i : p_i\theta \stackrel{*}{\leftrightarrow} q_i \theta$ ($i = 1,\ldots,m$).
(1) If $t \underset{\Gamma}{\sim} s$ 
then there exists a proof $\mathcal{Q}: t\theta \stackrel{*}{\leftrightarrow} s \theta$
with $w(\mathcal{Q}) \le \Sigma_{i=1}^m + k$
(2) If $t \underset{\Gamma}{\simto} s$ 
then there exists a proof $\mathcal{Q}: t\theta \to s \theta$
with $w(\mathcal{Q}) \le \Sigma_{i=1}^m + k + 1$.
\end{lemma}

Thus, Theorem 1  of \cite{WeightDec} follows for our Definition
of $\underset{\Gamma}{\sim}$ and $\underset{\Gamma}{\simto}$.

\begin{theorem}[Theorem 1  of \cite{WeightDec}, generalized]
Let $\mathcal{R}$ be a semi-equational non-duplicating LR-separated CTRS.
Then $\mathcal{R}$ is weight decreasing joinable if for any 
critical pair $\Gamma \vdash \langle s,t \rangle$ of $\mathcal{R}$,
either
(i) $s \underset{\Sigma}{\sim}  t$ for some $\Sigma \sqsubseteq \Gamma \sqcup \{ \bullet \}$,
(ii) 
$s \underset{\Sigma}{\simto} t$ or 
$t \underset{\Sigma}{\simto} s$ 
for some $\Sigma \sqsubseteq \Gamma \sqcup \{ \bullet \}$, or 
(iii) 
$s \underset{\Sigma_1}{\simto} \circ \underset{\Sigma_2}{\sim} t$
and $t \underset{\Sigma_1'}{\simto} \circ \underset{\Sigma_2'}{\sim} s$
for some $\Sigma_1,\Sigma_2,\Sigma_1',\Sigma_2'$
such that $\Sigma_1 \sqcup \Sigma_2 \sqsubseteq \Gamma \sqcup \{ \bullet \}$
and $\Sigma_1' \sqcup \Sigma_2' \sqsubseteq \Gamma \sqcup \{ \bullet \}$.
\end{theorem}

Below, we abbreviate $\{ \overbrace{\bullet, \ldots, \bullet}^{\text{$k$-times}} \}$ as $\{  \bullet^k \}$.

\begin{lemma}
\label{lem:correspondence I}
Let $\Lambda$ be a multiset of equations.
(i) 
If $\Lambda \Vdash_{\mathcal{R}} u \sim_k v$ 
then $u \underset{\Delta}{\sim} v$ for some
$\Delta = \Lambda' \sqcup \{ \bullet^k  \}$ 
such that $\Lambda' \sqsubseteq \Lambda$.
(ii) 
If $\Lambda \Vdash_{\mathcal{R}} u \to_k v$ 
then $u \underset{\Delta}{\simto} v$ for some 
$\Delta = \Lambda' \sqcup \{ \bullet^{k-1} \}$ 
such that $\Lambda' \sqsubseteq \Lambda$.
(iii)
If $\Lambda \Vdash_{\mathcal{R}} \langle u_1,\ldots,u_n \rangle  \sim_k \langle v_1, \ldots,v_n \rangle$
then 
$u_j \underset{\Delta_j}{\sim} v_j$ ($j = 1,\ldots,n$)
for some $\Delta_1,\ldots,\Delta_n$
such that $\bigsqcup_j \Delta_j = \Lambda' \sqcup \{ \bullet^k \}$
for some $\Lambda' \sqsubseteq \Lambda$.
\end{lemma}

\begin{proof}
The proofs of (i)--(iii) 
proceed by induction on the derivation
simultaneously.
\end{proof}

\newcommand{\eq}{\textit{eq}}

For any multiset $\Delta$ of equations and $\bullet$,
let $\Delta^\bullet$ be the multiset of $\bullet$
obtained from $\Delta$ by removing all equations,
and $\Delta^\eq$ be the multiset of equations
obtained from $\Delta$ by removing all $\bullet$.
Furthermore, we denote $\vert \Delta \vert$ the length of $\Delta$.

\begin{lemma}
\label{lem:correspondence II}
Let $\Delta$ be a multiset of equations and $\bullet$.
(i) 
If $u \underset{\Delta}{\sim} v$
then 
$\Lambda \Vdash_{\mathcal{R}} u \sim_k v$
for any $\Lambda \sqsupseteq \Delta^\eq$,
where $k = \vert \Delta^\bullet \vert$.
(ii) 
If $u \underset{\Delta}{\simto} v$ 
then 
$\Lambda \Vdash_{\mathcal{R}} u \sim_k v$
for any $\Lambda \sqsupseteq \Delta^\eq$,
where $k = \vert \Delta^\bullet \vert + 1$
(iii)
If $u_j \underset{\Delta_j}{\sim} v_j$ ($j = 1,\ldots,n$),
then 
$\Lambda \Vdash_{\mathcal{R}} \langle u_1,\ldots,u_n \rangle  \sim_k \langle v_1, \ldots,v_n \rangle$
for any $\Lambda \sqsupseteq \bigsqcup_j \Delta_j^\eq$,
where $k = \vert \bigsqcup_j \Delta_j^\bullet \vert$.
\end{lemma}

\begin{proof}
The proofs of (i)--(iii) 
proceed by induction on the derivation
simultaneously.
\end{proof}

\begin{lemma}
\label{lem:correspondence III}
Let $\Gamma$ be a multiset of equations.
(i) $s \underset{\Sigma}{\sim}  t$ for some $\Sigma \sqsubseteq \Gamma \sqcup \{ \bullet \}$
iff 
$\Gamma \Vdash_{\mathcal{R}} s \sim_{\le 1} t$.
(ii) 
$s \underset{\Sigma}{\simto} t$ for some $\Sigma \sqsubseteq \Gamma \sqcup \{ \bullet \}$
iff 
$\Gamma \Vdash_{\mathcal{R}} s  \to_1 t$ or $\Gamma \Vdash_{\mathcal{R}} s  \to_2 t$.
(iii) 
$s \underset{\Sigma_1}{\simto} \circ \underset{\Sigma_2}{\sim} t$
for some $\Sigma_1,\Sigma_2$ such that $\Sigma_1 \sqcup \Sigma_2 \sqsubseteq \Gamma \sqcup \{ \bullet \}$
iff 
$\Gamma \Vdash_{\mathcal{R}} s  \to_i \circ \sim_j t$ with $i + j \le 2$.
\end{lemma}

\begin{proof}
(i) ($\Rightarrow$)
Suppose
$s \underset{\Sigma}{\sim}  t$ for some $\Sigma \sqsubseteq \Gamma \sqcup \{ \bullet \}$.
Then by Lemma \ref{lem:correspondence II},
$\Lambda \Vdash_{\mathcal{R}} s \sim_k t$
for any $\Lambda \sqsupseteq \Sigma^\eq$,
where $k = \vert \Sigma^\bullet \vert$.
If $\Sigma \sqsubseteq \Gamma$
then $\bullet \notin \Sigma$,
and hence, $\Lambda \Vdash_{\mathcal{R}} s \sim_0 t$ for any $\Lambda \sqsupseteq \Sigma^\eq = \Sigma$,
as $k = \vert \Sigma^\bullet \vert = 0$.
Thus, 
$\Lambda \Vdash_{\mathcal{R}} s \sim_0 t$ for any $\Lambda \sqsupseteq \Sigma$.
Hence $\Gamma \Vdash_{\mathcal{R}} s \sim_0 t$.
Otherwise, we have $\bullet \in \Sigma$,
and hence, $\Sigma = \Sigma' \sqcup \{ \bullet \}$ for some $\Sigma' \sqsubseteq \Gamma$.
Then, 
$\Lambda \Vdash_{\mathcal{R}} s \sim_1 t$ for any $\Lambda \sqsupseteq \Sigma^\eq = \Sigma'$,
as $k = \vert \Sigma^\bullet \vert = 1$.
Thus, $\Gamma \Vdash_{\mathcal{R}} s \sim_1 t$.
Therefore, $\Gamma \Vdash_{\mathcal{R}} s \sim_{\le 1} t$ holds.
($\Leftarrow$)
Firstly, suppose $\Gamma \Vdash_{\mathcal{R}} s \sim_{0} t$.
Then, by Lemma \ref{lem:correspondence I},
$s \underset{\Sigma}{\sim} t$ for some
$\Sigma = \Gamma' \sqcup \{ \bullet^0  \}$ 
such that $\Gamma' \sqsubseteq \Gamma$,
i.e.\ 
$s \underset{\Sigma}{\sim} t$ for some $\Sigma \sqsubseteq \Gamma$.
Next, suppose $\Gamma \Vdash_{\mathcal{R}} s \sim_{1} t$.
Then, by Lemma \ref{lem:correspondence I},
$s \underset{\Sigma}{\sim} t$ for some
$\Sigma = \Gamma' \sqcup \{ \bullet^1  \}$ 
such that $\Gamma' \sqsubseteq \Gamma$,
i.e.\ 
$s \underset{\Sigma}{\sim} t$ for some $\Sigma \sqsubseteq \Gamma \sqcup \{ \bullet  \}$.
Thus, the claim holds.

\noindent
(ii) ($\Rightarrow$)
Suppose
$s \underset{\Sigma}{\simto}  t$ for some $\Sigma \sqsubseteq \Gamma \sqcup \{ \bullet \}$.
Then by Lemma \ref{lem:correspondence II},
$\Lambda \Vdash_{\mathcal{R}} s \to_k t$
for any $\Lambda \sqsupseteq \Sigma^\eq$,
where $k = \vert \Sigma^\bullet \vert + 1$.
If $\Sigma \sqsubseteq \Gamma$
then $\bullet \notin \Sigma$,
and hence, $\Lambda \Vdash_{\mathcal{R}} s \to_1 t$ for any $\Lambda \sqsupseteq \Sigma^\eq = \Sigma$,
as $k = \vert \Sigma^\bullet \vert + 1 = 1$.
Thus, 
$\Lambda \Vdash_{\mathcal{R}} s \to_1 t$ for any $\Lambda \sqsupseteq \Sigma$.
Hence $\Gamma \Vdash_{\mathcal{R}} s \to_1 t$.
Otherwise, we have $\bullet \in \Sigma$,
and hence, $\Sigma = \Sigma' \sqcup \{ \bullet \}$ for some $\Sigma' \sqsubseteq \Gamma$.
Then, 
$\Lambda \Vdash_{\mathcal{R}} s \to_2 t$ for any $\Lambda \sqsupseteq \Sigma^\eq = \Sigma'$,
as $k = \vert \Sigma^\bullet \vert + 1 = 2$.
Thus, $\Gamma \Vdash_{\mathcal{R}} s \to_2 t$.
Therefore, $\Gamma \Vdash_{\mathcal{R}} s \to_{1} t$
or $\Gamma \Vdash_{\mathcal{R}} s \to_{2} t$ holds.
($\Leftarrow$)
Firstly, suppose $\Gamma \Vdash_{\mathcal{R}} s \to_{1} t$.
Then, by Lemma \ref{lem:correspondence I},
$s \underset{\Sigma}{\simto} t$ for some
$\Sigma = \Gamma' \sqcup \{ \bullet^0  \}$ 
such that $\Gamma' \sqsubseteq \Gamma$,
i.e.\ 
$s \underset{\Sigma}{\simto} t$ for some $\Sigma \sqsubseteq \Gamma$.
Next, suppose $\Gamma \Vdash_{\mathcal{R}} s \to_{2} t$.
Then, by Lemma \ref{lem:correspondence I},
$s \underset{\Sigma}{\simto} t$ for some
$\Sigma = \Gamma' \sqcup \{ \bullet^1  \}$ 
such that $\Gamma' \sqsubseteq \Gamma$,
i.e.\ 
$s \underset{\Sigma}{\simto} t$ for some $\Sigma \sqsubseteq \Gamma \sqcup \{ \bullet  \}$.
Thus, the claim holds.

\noindent
(iii) ($\Rightarrow$)
Suppose
$s \underset{\Sigma_1}{\simto} \circ \underset{\Sigma_2}{\sim} t$
for some $\Sigma_1,\Sigma_2$ such that $\Sigma_1 \sqcup \Sigma_2 \sqsubseteq \Gamma \sqcup \{ \bullet \}$.
Firstly, if $\Sigma_1 \sqcup \Sigma_2 \sqsubseteq \Gamma$,
then,  as in the proof of (i) and (ii),
it follows $\Gamma \Vdash_{\mathcal{R}} s  \to_1 \circ \sim_0 t$.
Secondly, if $\bullet \in \Sigma_1$,
then as in the proof of (i) and (ii),
it follows $\Gamma \Vdash_{\mathcal{R}} s  \to_2 \circ \sim_0 t$.
Finally, if $\bullet \in \Sigma_2$,
then as in the proof of (i) and (ii),
it follows $\Gamma \Vdash_{\mathcal{R}} s  \to_1 \circ \sim_1 t$.
Thus, in any case, 
$\Gamma \Vdash_{\mathcal{R}} s  \to_i \circ \sim_j t$ with $i + j \le 2$.
($\Leftarrow$)
Suppose $\Gamma \Vdash_{\mathcal{R}} s  \to_i \circ \sim_j t$ with $i + j \le 2$.
Then we have cases 
(a) $\Gamma \Vdash_{\mathcal{R}} s  \to_1 u \sim_0 t$,
(b) $\Gamma \Vdash_{\mathcal{R}} s  \to_1 u  \sim_1 t$, and
(c) $\Gamma \Vdash_{\mathcal{R}} s  \to_2 u  \sim_0 t$.
In case (a),
there exist $\Gamma_1, \Gamma_2$ such that 
$\Gamma = \Gamma_1 \sqcup \Gamma_2$,
$\Gamma_1 \Vdash_{\mathcal{R}} s  \to_1 u$ and
$\Gamma_2 \Vdash_{\mathcal{R}} u  \sim_0 t$.
Then,  as in the proof of (i) and (ii),
$s \underset{\Sigma_1}{\simto} u$ for some for some $\Sigma_1 \sqsubseteq \Gamma_1$
and
$u \underset{\Sigma_2}{\sim} t$ for some for some $\Sigma_2 \sqsubseteq \Gamma_2$.
In case (b), similarly, we have 
$s \underset{\Sigma_1}{\simto} u$ for some for some $\Sigma_1 \sqsubseteq \Gamma_1$
and
$u \underset{\Sigma_2}{\sim} t$ for some for some $\Sigma_2 \sqsubseteq \Gamma_2  \sqcup \{ \bullet \}$.
In case (c), similarly, we have 
$s \underset{\Sigma_1}{\simto} u$ for some for some $\Sigma_1 \sqsubseteq \Gamma_1 \sqcup \{ \bullet \}$.
and
$u \underset{\Sigma_2}{\sim} t$ for some for some $\Sigma_2 \sqsubseteq \Gamma_2$.
Thus, the claim holds.
\end{proof}

\begin{proof}[Proof of Theorem \ref{thm:I}]
It follows immediately from Lemma~\ref{lem:correspondence III},
by noting 
$\Gamma \Vdash_{\mathcal{R}} s \to_{1} t$
implies $\Gamma \Vdash_{\mathcal{R}} s \sim_{1} t$.
\end{proof}

\newpage

\section{Some detailed proofs}

\begin{proof}[Proof of Lemma~\ref{lem:correspondence I}]
We prove (i)--(iii) simultaneously by induction on the derivation.
\begin{enumerate}
\item Case $\Gamma \sqcup \{ u  \approx v \} \Vdash_{\mathcal{R}} u \sim_0 v$.
The claim holds since $u \underset{\{ u \approx v \}}{\sim} v$ by \textbf{asp}.

\item Case $\Gamma \Vdash_{\mathcal{R}} t \sim_0 t$
The claim holds since $t \underset{\{  \}}{\sim} t$ by \textbf{refl}.

\item Case $\Gamma \Vdash_{\mathcal{R}} s \sim_i t $ is derived from $\Gamma \Vdash_{\mathcal{R}} t \sim_i s$.
By induction hypothesis, 
$t  \underset{\Delta}{\sim} s$ for some 
$\Delta = \Gamma' \sqcup \{ \bullet^i \}$ 
such that $\Gamma' \sqsubseteq \Gamma$.
Then $s  \underset{\Delta}{\sim} t$ by \textbf{sym}, and the claim holds.

\item Case $\Gamma \sqcup \Sigma \Vdash_{\mathcal{R}} s \sim_{i+j} u$
is derived from $\Gamma \Vdash_{\mathcal{R}} s \sim_i t$ and $\Sigma \Vdash_{\mathcal{R}} t \sim_j u$
By induction hypothesis, 
$s  \underset{\Delta_1}{\sim} t$ for some 
$\Delta_1 = \Gamma' \sqcup \{ \bullet^i \}$ 
such that $\Gamma' \sqsubseteq \Gamma$,
and 
$t  \underset{\Delta_2}{\sim} u$ for some 
$\Delta_2 = \Sigma' \sqcup \{ \bullet^j \}$ 
such that $\Sigma' \sqsubseteq \Sigma$.
Take $\Delta = \Delta_1 \sqcup \Delta_2$.
Then $s  \underset{\Delta}{\sim} u$ by \textbf{trans}.
Furthermore, 
$\Delta = \Delta_1 \sqcup \Delta_2
= \Gamma' \sqcup \{ \bullet^i \}
\sqcup \Sigma' \sqcup \{ \bullet^j \}
= \Gamma' \sqcup \Sigma' \sqcup \{ \bullet^{i+j} \}$,
and $\Gamma' \sqcup \Sigma' \sqsubseteq \Gamma \sqcup  \Sigma$.
Hence the claim holds.

\item Case $\Gamma \Vdash_{\mathcal{R}} C[s] \sim_i C[t]$ is derived from $\Gamma \Vdash_{\mathcal{R}} s \sim_i t$
By induction hypothesis, 
$s  \underset{\Delta}{\sim} t$ for some 
$\Delta = \Gamma' \sqcup \{ \bullet^i \}$ 
such that $\Gamma' \sqsubseteq \Gamma$.
Then $C[s]  \underset{\Delta}{\sim} C[t]$ by \textbf{cntxt}, and the claim holds.

\item Case 
$\bigsqcup_j \Gamma_j \Vdash_{\mathcal{R}} \langle u_1,\ldots,u_n \rangle  \sim_k \langle v_1, \ldots,v_n \rangle$
is derived from 
$\Gamma_1 \Vdash_{\mathcal{R}} u_1 \sim_{i_1} v_1, \ldots, \Gamma_n \Vdash_{\mathcal{R}} u_n \sim_{i_n} v_n$
where $k = \sum_j i_j$.
By induction hypothesis,  for each $j = 1,\ldots,n$,
$u_j  \underset{\Delta_j}{\sim} v_j$ for some 
$\Delta_j = \Gamma_j' \sqcup \{ \bullet^{i_j} \}$ 
such that $\Gamma_j' \sqsubseteq \Gamma_j$.
Since $\bigsqcup_j \Delta_j = 
\bigsqcup_j \Gamma_j' \sqcup \{ \bullet^k \}$
and $\bigsqcup_j \Gamma_j' \sqsubseteq \bigsqcup_j \Gamma_j$,
the claim holds.

\item Case 
$\Gamma \Vdash_{\mathcal{R}} s \sim_i t$ is derived from $\Gamma \Vdash_{\mathcal{R}} s \to_i t$.
By induction hypothesis, 
$s  \underset{\Delta}{\simto} t$ for some 
$\Delta = \Gamma' \sqcup \{ \bullet^{i-1} \}$ 
such that $\Gamma' \sqsubseteq \Gamma$.
Then $s  \underset{\Delta \sqcup \{ \bullet \}}{\sim} t$ by \textbf{bullet}
and $\Delta \sqcup \{ \bullet \} = \Gamma' \sqcup \{ \bullet^i \}$.
Thus, the claim holds.

\item Case 
$\Gamma \Vdash_{\mathcal{R}} C[l\sigma] \to_{i+1}  C[r\sigma]$
is derived from 
$\Gamma \Vdash_{\mathcal{R}} \langle x_1\sigma,\ldots,x_n\sigma \rangle  \sim_i \langle y_1\sigma, \ldots,y_n\sigma \rangle$
where $l \to r \Leftarrow x_1 \approx y_1,\ldots,x_n \approx y_n \in \mathcal{R}$.
By induction hypothesis, 
$x_j\sigma \underset{\Delta_j}{\sim} y_j\sigma$ ($j = 1,\ldots,n$)
for some $\Delta_1,\ldots,\Delta_n$
such that $\bigsqcup_j \Delta_j = \Gamma' \sqcup \{ \bullet^i \}$
for some $\Gamma' \sqsubseteq \Gamma$.
Then, by \textbf{rule},
we have $C[l\theta] \underset{\Delta}{\simto} C[r\theta]$ 
where $\Delta= \bigsqcup_j \Delta_j$.
Thus, the claim holds.
\end{enumerate}
\end{proof}

\begin{proof}[Proof of Lemma~\ref{lem:correspondence II}]
We prove (i)--(iii) simultaneously by induction on the derivation.

\begin{enumerate}
\item Case \textbf{(asp)}.
We have $t \underset{\{ t \approx s \}}{\sim} s$.
Then 
$\Lambda \Vdash_{\mathcal{R}} t \sim_0 s$
for any $\Lambda \sqsupseteq \{ t \approx s \}$ by definition.

\item Case \textbf{(refl)}.
We have $t \underset{\{   \}}{\sim} t$.
Then 
$\Lambda \Vdash_{\mathcal{R}} t \sim_0 t$ for any $\Lambda$ by definition.

\item Case \textbf{(sym)}.
Suppose $s \underset{\Gamma}{\sim} t$ is derived from $t \underset{\Gamma}{\sim} s$.
Let $\Lambda \sqsupseteq \Gamma^\eq$.
Then by induction hypothesis, 
$\Lambda \Vdash_{\mathcal{R}} t \sim_k s$, where $k = \vert \Gamma^\bullet \vert$.
Then, it follows $\Lambda \Vdash_{\mathcal{R}} t \sim_k s$ by definition.

\item Case \textbf{(trans)}.
Suppose 
$t \underset{\Gamma \sqcup \Sigma}{\sim} s$
is derived from $t \underset{\Gamma}{\sim} r$ and $r \underset{\Sigma}{\sim} s$.
Let $\Lambda \sqsupseteq (\Gamma \sqcup \Sigma)^\eq$.
Then, there exist
$\Lambda_1,\Lambda_2$
such that 
$\Lambda = \Lambda_1 \sqcup \Lambda_2$,
$\Lambda_1 \sqsupseteq \Gamma^\eq$
and 
$\Lambda_2 \sqsupseteq \Sigma^\eq$.
Then, by induction hypothesis, 
$\Lambda_1 \Vdash_{\mathcal{R}} t \sim_{k_1} s$
where $k_1 = \vert \Gamma^\bullet \vert$,
and 
$\Lambda_2 \Vdash_{\mathcal{R}} t \sim_{k_2} s$
where $k_2 = \vert \Sigma^\bullet \vert$.
Then, it follows $\Lambda \Vdash_{\mathcal{R}} t \sim_{k_1+k_2} s$ by definition.
As $k_1+ k_2 = \vert \Gamma^\bullet \vert + \vert \Sigma^\bullet \vert
= \vert (\Gamma \sqcup \Sigma)^\bullet \vert$,
the claim follows.

\item Case \textbf{(cntxt)}.
Suppose
$C[t] \underset{\Gamma}{\sim} C[s]$ is derived from $t \underset{\Gamma}{\sim} s$.
Let $\Lambda \sqsupseteq \Gamma^\eq$.
Then by induction hypothesis, 
$\Lambda \Vdash_{\mathcal{R}} t \sim_k s$, where $k = \vert \Gamma^\bullet \vert$.
Then, it follows $\Lambda \Vdash_{\mathcal{R}} C[t] \sim_k C[s]$ by definition.

\item Case \textbf{(rule)}.
Suppose $C[l\theta] \underset{\Gamma}{\simto} C[r\theta]$
is derived from 
$x_1\theta \underset{\Gamma_i}{\sim} y_i\theta$  $(i = 1,\ldots,n)$,
where $\Gamma = \Gamma_1 \sqcup \cdots \sqcup \Gamma_n$
and $l \to r \Leftarrow x_1 \approx y_1,\ldots,x_n \approx y_n \in \mathcal{R}$.
Let $\Lambda \sqsupseteq \Gamma^\eq$.
Then, there exist
$\Lambda_1,\ldots,\Lambda_n$
such that 
$\Lambda = \bigsqcup_j \Lambda_j$
and $\Lambda_j \sqsupseteq \Gamma_j^\eq$ for each $1 \le j  \le n$.
Hence, by induction hypothesis,
$\Lambda_j \Vdash_{\mathcal{R}} x_j\theta  \sim_{k_j} y_j \theta$
where $k_j = \vert \Gamma_j^\bullet \vert$ for each $1 \le j  \le n$.
Then, by definition,
$\Lambda \Vdash_{\mathcal{R}} \langle x_1\theta,\ldots,x_n\theta \rangle  \sim_{k'} \langle y_1\theta, \ldots,y_n\theta \rangle$
where $k' = \sum_j k_j = \sum_j \vert \Gamma_j^\bullet \vert = \vert (\bigsqcup_j \Gamma_j)^\bullet \vert
= \vert \Gamma^\bullet \vert$.
Then, by definition,
$\Lambda \Vdash_{\mathcal{R}} C[l\theta] \sim_{k'+1} C[r\theta]$.

\item Case \textbf{(bullet)}.
Suppose $s \underset{\Gamma \sqcup \{ \bullet \}}{\sim} t$
is derived from $t \underset{\Gamma}{\simto} s$.
Let $\Lambda \sqsupseteq \Gamma^\eq$.
Then by induction hypothesis, 
$\Lambda \Vdash_{\mathcal{R}} s \simto_k s$, where $k = \vert \Gamma^\bullet \vert + 1$.
Then, it follows $\Lambda \Vdash_{\mathcal{R}} t \sim_k s$ by definition.

\end{enumerate}
\end{proof}

\end{document}